\documentclass[article]{article}
\usepackage{color}
\usepackage{graphicx}
\usepackage{amsmath}
\usepackage{amssymb}
\usepackage{mathrsfs}
\usepackage[numbers,sort&compress]{natbib}
\usepackage{amsthm}
\numberwithin{equation}{section}
\usepackage{pgf}
\usepackage{CJK}
\usepackage{graphicx}
\usepackage{tikz}
\usepackage{stmaryrd}
\usepackage{graphicx}
\usepackage{hyperref}
\usepackage[latin1]{inputenc}
\usepackage[T1]{fontenc}
\usepackage[english]{babel}
\usepackage{listings}
\usepackage{xcolor,mathrsfs,url}
\usepackage{amssymb}
\usepackage{amsmath}
\usepackage{ifthen}

\newtheorem{proposition}{Proposition}
\usepackage {caption}
\usepackage{float}
\usepackage{subfigure}

\DeclareMathOperator*{\res}{Res}

\begin{document}

\title{ Riemann-Hilbert approach  to the modified nonlinear Schr\"{o}dinger  equation with non-vanishing asymptotic  boundary conditions }
\author{Yiling YANG$^1$ and Engui FAN$^{1}$\thanks{\ Corresponding author and email address: faneg@fudan.edu.cn } }
\footnotetext[1]{ \  School of Mathematical Sciences, Fudan University, Shanghai 200433, P.R. China.}

\date{ }

\maketitle
\begin{abstract}
	\baselineskip=17pt
	
	The modified nonlinear Schr\"{o}dinger (NLS) equation  was proposed to describe the nonlinear propagation of the Alfven waves and the femtosecond optical pulses in a nonlinear single-mode optical fiber.  In this paper, the inverse scattering transform for the modified  NLS  equation
	with  non-vanishing  asymptotic  boundary at infinity
	is presented. An appropriate two-sheeted Riemann surface is introduced to map the original spectral parameter $k$ into a single-valued parameter $z$.
	The    asymptotic behaviors, analyticity and the symmetries of the Jost solutions of Lax pair for the modified NLS equation, as well as  the spectral  matrix are   analyzed in details.
	Then  a  matrix Riemann-Hilbert (RH) problem associated with the problem of nonzero asymptotic boundary conditions is established, from which
	$N$-soliton solutions is obtained via   the corresponding reconstruction formulae.
	As an illustrate examples of $N$-soliton formula,    two kinds of one-soliton solutions and  three kinds of two-soliton solutions are
	explicitly presented according to different distribution of the spectrum.  The  dynamical feature  of those solutions  are
	characterized   in the particular case with a quartet of discrete eigenvalues.  It is shown  that distribution of the spectrum and non-vanishing   boundary also affect
	feature  of soliton solutions. Finally, we analyze the differences between  our results  and  those    on    zero boundary case.   
\end{abstract}

\baselineskip=17pt

\newpage

\section {Introduction}

\quad
\quad The nonlinear  NLS  equation \cite{RN1,RN2}
\begin{equation}
iu_t+u_{xx}+2|u|^2u=0\label{NLS}
\end{equation}
is one of the most celebrated soliton equations, which has applications in a wide variety of fields such as plasma physics, nonlinear optics and many other fields.
To study the effect of higher-order perturbations,
various modifications and generalizations of the NLS equations have been proposed and studied
\cite{RN3,RN4}.    Among them, there are three celebrated  derivative nonlinear Schr\"{o}dinger   equations, including    the Kaup-Newell equation
\cite{Kaup},  the Chen-Lee-Liu  equation \cite{CLL} and the Gerdjikov-Ivanov equation \cite{GI,Fan1, Fan2}.
It is known that these three equations may be
transformed into each other by   implicit  gauge transformations, and the method of gauge transformation
can also be applied to some generalized cases \cite{RN15, Fan3}.
1970s, Wadati et al proposed a   kind of    mixed   NLS   equation
\begin{equation}
u_t-iu_{xx}+a (|u|^2u)_x+i b |u|^2u=0,\label{mixedNLS}
\end{equation}
which was shown to be completely integrable by  inverse scattering transformation  \cite{RN5}.  For $a=0$, the equation (\ref{mixedNLS}) reduces the classical NLS equation (\ref{NLS});
For  $b=0$, the equation (\ref{mixedNLS})  reduces  to  the Kaup-Newell    equation;
For $a, b  \not=0$,   the equation (\ref{mixedNLS}) is equivalent  to   the modified  NLS equation
\begin{equation}
iu_t+u_{xx}+2|u|^2u+i\frac{1}{\alpha}(|u|^2 u)_x=0, \ \ \alpha>0, \label{MNLS}
\end{equation}	
which   was also called the perturbation NLS equation \cite{Maimistov1},  it  can   be used
to describe Alfven waves propagating along the magnetic field
in cold plasmas   and the deep-water gravity waves \cite{Mio1,Stiassnie1}.   The term $i(|u|^2u)_x$ in the  equation (\ref{MNLS})  is called the self-steepening term, which
causes an optical pulse to become asymmetric and steepen upward at the trailing edge
\cite{Al, Yang}.  The equation (\ref{MNLS}) also describes the short pulses propagate in a long optical fiber characterized by a nonlinear refractive index \cite{Nakatsuka1,Tzoar1}.
Brizhik et al  showed  that the modified NLS equation (\ref{MNLS}), unlike the
classical  NLS equation (\ref{NLS}), possesses static localized solutions when the effective nonlinearity parameter is larger than a certain critical value \cite{RN13}.

The modified NLS  equation (\ref{MNLS}) has   been discussed extensively,  for example,
various  solutions such as analytical solutions,  soliton solutions, rational  and multi-rogue wave solutions were found   by   analytical method,  Hirota
bilinear method and   Darboux transformation respectively \cite{RN9, yangxiao, RN59, RN11}.   The Hamiltonian structure  for the equation (\ref{MNLS}) was given \cite{RN10}.
$N$-soliton solutions for  the modified NLS equation (\ref{MNLS}) with zero boundary condition $u(x,t)\rightarrow 0,\ x\to{\pm}\infty$  also were obtained by inverse   scattering  transform (IST)  and dressing method \cite{chen1990,chen1991,Doktorov}.   Deift-Zhou nonlinear  steepest decedent  method  was used to  obtain  long-time  asymptotic solution   of initial problem of   the equation (\ref{MNLS})   \cite{RN12}.   In recent years,     coupled modified  NLS  equations and
vector modified NLS equations  also were   presented and studied  \cite{RN8,NR3,NR4,NR5,NR1}.

Solitons are found in various areas of physics such as   gravitation and field theory, plasma physics,   nonlinear optics and solid state physics, which can be described  by nonlinear equations.
The  IST  procedure,  as one of the  most powerful tool to  investigate solitons of  nonlinear models,   was first discovered  by Gardner,  Green,
Kruskal   and Miura \cite{GGKM1967}.   The  IST  for the focusing NLS equation  with zero boundary conditions   was first developed by Zakharov
and Shabat \cite{Zakharov1972},  later    for the defocusing case with nonzero boundary conditions  \cite{Zakharov1973}.
The next important steps of the development of  IST method   is   the Riemann-Hilbert (RH) method   as the  modern version of IST  was   established  by Zakharov and Shabat   \cite{Zakharov1979},
which involves the determination of a analytic  function in given sectors of the complex plane, from the knowledge of the jumps
of this function across the boundaries of the  sectors.
It has since become clear that the    RH  method  is applicable to  construction  of exact solutions   and  asymptotic  analysis of  solutions
for   a wide class of    integrable systems  \cite{Prinari,F1,F2,F3,F4,F5,F6,ZF2019,ZF20192}.

To our knowledge, with the exception of  IST and dressing method  to the modified  NLS  equation with zero boundary case  \cite{chen1990,chen1991,Doktorov},
there are almost no known results on IST or RH method  for the  modified  NLS  equation with nonzero boundary conditions.  However, in many laboratory and field situations,
the wave motion is initiated by what corresponds to the imposition of boundary conditions.  In addition,  the modulational instability  has received renewed
interest in recent years, and has also been suggested as a possible mechanism for the generation of
rogue waves \cite{cos2006,Zo2009}.  It was shown that rogue wave solutions of integrable systems can be obtained via IST or RH method,  which  provide   an  effective  and perfect  tool   to study  rogue waves and  the nonlinear stage of modulational instability  \cite{F3,dp2019,zc2019}.

In this article,  we consider the modified NLS equation (\ref{MNLS}) with the following nonzero asymptotic  boundary conditions
\begin{equation}
u(x,t) \sim  u_{\pm}e^{-4i\alpha^2 t+2i\alpha x},\hspace{0.5cm}x\to{\pm}\infty,\label{bc1}
\end{equation}
where $ \left|u_{\pm}\right|= u_0 > 0$, and $u_{\pm}$ are independent of $x, t$.
Our aim here is,  by using  Riemann-Hilbert (RH) method,     to   establish
a formulae  of   $N$-soliton solutions   for the   above nonzero boundary   problem   and   characterize  their  features  in the particular case with a quartet of discrete eigenvalues.

The structure of this work is the following.  In  section 2,  we introduce an appropriate two-sheeted Riemann surface is introduced to
map the original spectral parameter $k$ into a single-valued parameter $z$.
In Section 3 and Section 4,  starting from the Lax pair of the  modified NLS equation,  we  construct Jost solutions and spectral matrix.
Then in Section 5,   we   analyze    analytical and  symmetric  properties  the Jost solutions and  spectral matrix.
In section 6,  we  analyze   asymptotic behaviors of the Jost solution, scattering matrix and reflection coefficients.
In Section 7, we discuss the discrete spectrum and the residue conditions to analyze  poles for meromorphic matrices  appearing in  the RH problem.
In Section 8, we  establish reconstruction  formula   between solution of the modified NLS equation  and the RH  problem.
We obtain the trace formula as well as theta condition that reflection coefficients and discrete spectrum satisfy.
In the section 9,   in  reflectionless case,  we discuss  solvability of the RH problem,  from which  $N$-soliton solutions of the modified NLS equation are  obtained.
As an illustrate examples of $N$-soliton formula,   according to different distribution of the spectrum,  two kinds of one-soliton solutions and  three kinds of two-soliton solutions are
explicitly presented,  and their dynamical features   are  characterized     with a quartet of discrete eigenvalues.
The affects   of   distribution of the spectrum  on    soliton solutions are analyzed.

\section{Riemann surface and uniformization variable}

\quad The modified NLS equation (\ref{MNLS}) admits the Lax pair \cite{Doktorov}
\begin{equation}
\phi_x = U\phi,\hspace{0.5cm}\phi_t = V\phi, \label{lax1}
\end{equation}
where
\begin{equation}
U=-\alpha i (k^2-1)\sigma_3+ikQ,\nonumber
\end{equation}
\begin{equation}
V=-2i\alpha^2(k^2-1)^2\sigma_3+2i\alpha k(k^2-1)Q+ik^2Q^2\sigma_3-k\sigma_3Q_x-\frac{i}{\alpha}kQ^3,\nonumber
\end{equation}
and
\begin{equation}
\hspace{0.5cm}\sigma_3=\left(\begin{array}{cc}
1 & 0   \\
0 & -1
\end{array}\right),\hspace{0.5cm}Q=\left(\begin{array}{cc}
0 & u  \\
u^{\ast} &0
\end{array}\right).\nonumber
\end{equation}

To change  (\ref{bc1})  into   constant boundary,  we make a   transformation
\begin{align*}
&u \rightarrow ue^{-4i\alpha^2t+2i\alpha x},\\
&\phi \rightarrow e^{(-2i\alpha^2t+i\alpha x)\sigma_3}\phi,
\end{align*}
then the  modified   NLS equation (\ref{MNLS})   becomes
\begin{equation}
iu_t+u_{xx}+4i\alpha u_x+i\frac{1}{\alpha} (|u|^2 u)_x=0\label{MNLS1}
\end{equation}
with corresponding boundary
\begin{equation}
\lim_{x \to \pm\infty}u(x,t)=u_\pm, \label{boundary}
\end{equation}
where $u_\pm$ are constant independent of $x, t$, and  $|u_\pm|=u_0$.  And the Lax pair (\ref{lax1}), as  the compatibility   of modified   NLS equation  (\ref{MNLS1}),
is changed to
\begin{equation}
\phi_x=X\phi,\hspace{0.5cm}\phi_t=T\phi,\label{lax2}
\end{equation}
where
\begin{align*}
X=-i\alpha k^2\sigma_3+ikQ,\hspace{0.5cm}Q=\left(\begin{array}{cc}
0 & u \\
u^* & 0
\end{array}\right),
\end{align*}
\begin{align*}
T=(-2i\alpha^2 k^4+4i\alpha^2 k^2)\sigma_3+ik^2|u|^2\sigma_3+2i\alpha k(k^2-2)Q-\frac{i}{\alpha}kQ^3+kQ_x.
\end{align*}

Under the nonzero asymptotic  boundary  condition  (\ref{boundary}),  limit  spectral problem of  the Lax pair (\ref{lax2}) is
\begin{equation}
\psi_x=X_\pm\psi,\hspace{0.5cm}\psi_t=T_\pm\psi,\label{asymptoticlax}
\end{equation}
where
\begin{equation}
X_\pm=-i\alpha k^2\sigma_3+ikQ_\pm,\hspace{0.5cm}T_\pm=\left(2\alpha k^2-4\alpha-\frac{u_0^2}{\alpha}\right)X_\pm,
\end{equation}
and
\begin{equation*}
Q_\pm=\left(\begin{array}{cc}
0 & u_\pm \\
u_\pm^* & 0
\end{array}\right).
\end{equation*}

The eigenvalues of the matrix $X_\pm$ are $\pm ik\lambda$,
where $\lambda$ satisfies
\begin{equation}
\lambda^2=\alpha^2k^2+u_0^2,
\end{equation}
which  are doubly branched,  and its  branch points are  $k=\pm iu_0/\alpha$.
 Gluing  two copies of the complex plane $S_1$ and $S_2$ along the segment $[-\frac{i}{\alpha}u_0,\frac{i}{\alpha}u_0]$, we then obtain a  Riemann surface.
By setting
$$ k\alpha+iu_0=r_1e^{i\theta_1}, \ \   k\alpha-iu_0=r_2e^{i\theta_2},\ \ -\pi/2 < \theta_j <3/2\pi, \  j=1,2, $$
 we  get   two single-valued analytic functions on the Riemann surface
\begin{equation}
\lambda(k)=\Bigg\{\begin{array}{ll}
\text{$\frac{1}{\alpha}(r_1r_2)^{1/2}e^{i(\theta_1+\theta_2)/2},$} &\text{on $S_1,$}\\\\
\text{$-\frac{1}{\alpha}(r_1r_2)^{1/2}e^{i(\theta_1+\theta_2)/2},$} &\text{on $S_2$}. \end{array}
\label{tranlambda}
\end{equation}

To avoid multi-valued case  of   eigenvalue  $\lambda$,   we  introduce    a uniformization variable
\begin{equation}
z=\alpha k+\lambda,
\end{equation}
and  obtain two single-valued functions
\begin{equation}
k(z)=\frac{1}{2\alpha}(z-\frac{u_0^2}{z}),\hspace{0.5cm}\lambda(z)=\frac{1}{2}(z+\frac{u_0^2}{z}),\label{uniformization55}
\end{equation}
which  allow us  to discuss the IST  on a standard $z$-plane instead of the more cumbersome two-sheeted Riemann surface.  The  second transformation of  (\ref{uniformization55})
is   well-known   Joukowsky transformation.  However,
 the  limit    $k\rightarrow \infty$  on   two-sheeted Riemann surface will  cause   two limits $z$-plane:   $z \rightarrow 0$ and
$ z \rightarrow \infty$.  In fact, for $k\in S_1$, we have
\begin{align}
z&= \alpha k+\sqrt{\alpha^2k^2+u_0^2}=\alpha k+\alpha k\left(1+ \frac{u_0^2}{\alpha^2k^2} \right)^{1/2}\nonumber\\
&=2\alpha k+O(k^{-1}) \rightarrow \infty,\hspace{0.5cm} \ k\rightarrow \infty.\nonumber
\end{align}
While for $k\in S_{2}$,  we find that
\begin{align}
& z=\alpha k-\sqrt{\alpha^2k^2+q_0^2}=\frac{-u_0^2}{\alpha k+\sqrt{\alpha^2k^2+u_0^2}} \rightarrow 0, \hspace{0.5cm}\ k\rightarrow \infty.\nonumber
\end{align}
 Therefore we should  consider  two kinds of asymptotic behaviors  of Jost solution,  scattering data and Riemann-Hilbert problem   on the $z$-plane  hereafter.

In addition to, we should explain  how the  two-sheeted Riemann surface is mapped  into  the  $z$-plane.
 The transformation (\ref{tranlambda}) possesses the following properties
\begin{itemize}
\item[$\blacktriangleright$]  Map ${\rm Im}k>0$ of  $S_1$ and  ${\rm Im}k<0$  of   $S_2$  together   into  the  ${\rm Im}\lambda>0$ of $\lambda$-plane;

\item[$\blacktriangleright$]  Map ${\rm Im}k<0$ of  $S_1$ and ${\rm Im}k>0$ of  $S_2$  together   into  the ${\rm Im}\lambda<0$ of $\lambda$-plane;

\item[$\blacktriangleright$]   Map the   segment $[-\frac{i}{\alpha}u_0,\frac{i}{\alpha}u_0]$       into  a  $[-u_0, u_0]$ on $\lambda$-plane.
\end{itemize}

Let's consider map from the $\lambda$-plane to the $z$-plane again.   By using  the  relation
\begin{align}
& \lambda(z)=\frac{z^2+u_0^2 }{2z}=\frac{(|z|^2-u_0^2)z+u_0^2(z+\bar{z}) }{2|z|^2} \nonumber\\
&= \frac{1}{2|z|^2} \left[(|z|^2-u_0^2)z+2u_0^2{\rm Re}z \right], \nonumber
\end{align}
we have
\begin{align}
& {\rm Im} \lambda(z)=  \frac{1}{2|z|^2}  (|z|^2-u_0^2)  {\rm Im} z,
\end{align}
which implies that  the  Joukowsky transformation admits the following properties
\begin{itemize}
\item[$\blacktriangleright$]   Map    ${\rm Im}\lambda>0$  into the domain $\{z\in C: (|z|^2-u_0^2){\rm Im} z>0\}$ in $z$-plane;

\item[$\blacktriangleright$]   Map  ${\rm Im}\lambda<0$  into the domain $\{z\in C: (|z|^2-u_0^2){\rm Im} z<0\}$ in $z$-plane;

\item[$\blacktriangleright$]  Map   $[-u_0, u_0]$  into circle  $C_0=\{|z|=u_0, \  z\in C\} $  in  $z$-plane.

\end{itemize}
Transformation relations  from $k$ two-sheeted Riemann surface,  $\lambda$-plane and  $z$-plane are shown in  Figure   \ref{klabdaz}.
It will be seen later that the circle $C_0$ is not   the boundary of analytical domains for  the  Jost  solutions and  scattering data,
 but  it  will  affect  their symmetry.  These are very   different from  classical focusing  NLS equation with nonzero boundary conditions,
  where  the boundary of analytical domains  is   $\mathbb{R}\cup C_0 $   \cite{F3}, but ours is $ \mathbb{R} \cup i\mathbb{R}$.

\begin{figure}
\begin{center}
\begin{tikzpicture}
\draw [fill=yellow,ultra thick,yellow] (-3.8,0) rectangle (-0.8,1);
\draw [fill=cyan,ultra thick,cyan] (-3.8,0) rectangle (-0.8,-1);
\draw [->](-4,0)--(-0.6,0);
\draw [->](-2.2,-1.5)--(-2.2,1.5);
\draw [ultra thick,red](-2.2,-0.6)--(-2.2,0.6);
\node at (-2.6, -0.6 )  {$-\frac{iu_0}{\alpha} $};
\node at (-2.5, 0.6 )  {$\frac{iu_0}{\alpha} $};
\node at (-3.4, 0.7 )  {$  S_1$};
\node [thick] [above]  at (-1.5, 0.3){\footnotesize ${\rm Im}k>0$};
\node [thick] [above]  at (-1.5, -0.8){\footnotesize ${\rm Im}k<0$};
\node [thick] [above]  at (-2.2,1.5){\footnotesize ${\rm Im}k$};

\draw [fill=yellow,ultra thick,yellow] (0.6,0) rectangle (3.5,-1);
\draw [fill=cyan,ultra thick,cyan] (0.6,0) rectangle (3.5,1);
\draw [-> ] (0.5,0)--(3.7,0);
\draw [->](2,-1.5)--(2,1.5);
\draw [ultra thick,red](2,-0.6)--(2,0.6);
\node at (1.6, -0.6 )  {$-\frac{iu_0}{\alpha} $};
\node at (1.7,  0.6 )    {$\frac{iu_0}{\alpha}$};
\node [thick] [above]  at (2.8, -0.8){\footnotesize ${\rm Im}k<0$};
\node [thick] [above]  at (2.8, 0.3){\footnotesize ${\rm Im}k>0$};
\node [thick] [above]  at (1, 0.5){\footnotesize $S_2$};
\node [thick, red] [above]  at (0,-0.2) {\footnotesize {\bf +}};
\draw [thick, red,->] (0,-0.7)--(-1.5,-2 );
\draw [thick, red,->] (0,-0.7)--(1.5,-2 );
\node [thick] [above]  at (-1.7,-2.1) {\footnotesize $\lambda=\sqrt{k^2+u_0^2}$};
\node [thick] [above]  at (2.2,-2.1)  {\footnotesize $k(z)=\frac{1}{2\alpha}(z-u_0^2/z)$};
\node [thick] [above]  at (2,1.5){\footnotesize ${\rm Im}k$};

\draw [fill=yellow,ultra thick,yellow] (-3.8,-4) rectangle (-0.8,-3 );
\draw [fill=cyan,ultra thick,cyan] (-3.8,-4) rectangle (-0.8,-5 );
\draw [->](-4,-4)--(-0.6,-4);
\draw [->](-2.2,-5.2)--(-2.2,-2.5);
\draw [ultra thick,red](-2.8,-4)--(-1.6,-4);
\node at (-2.8, -4.2  )  {$-u_0 $};
\node at (-1.6, -4.2 )  {$u_0 $};
\node [thick] [above]  at (-0.3,-4.2){\footnotesize ${\rm Re}\lambda$};
\node [thick] [above]  at (-2.2,-2.6){\footnotesize ${\rm Im}\lambda$};
\node [thick] [above]  at (-1.5,-3.8){\footnotesize ${\rm Im}\lambda>0$};
\node [thick] [above]  at (-1.5,-4.8){\footnotesize ${\rm Im}\lambda<0$};
\draw [fill=cyan,ultra thick,cyan] (0.8, -4) rectangle (3.5, -5 );
\path [fill=yellow] (1.4,-4) to  [out=-90, in=180]  (2,-4.6) to [out=0, in=-90]  (2.6,-4);
\draw [fill=yellow,ultra thick,yellow] (0.8, -4) rectangle (3.5, -3 );
\path [fill=cyan]  (1.4,-4) to  [out=90, in=180]  (2,-3.4) to [out=0, in=90]  (2.6,-4);

\draw [ultra thick,red] (2,-4) circle [radius=0.6];

\node [thick] [above]  at (4,-4.2){\footnotesize ${\rm Re}z$};
\node [thick] [above]  at (2,-2.6){\footnotesize ${\rm Im}z$};
\draw [thick, red,->] (-1,-4.8)--(1,-4.8 );
\node [thick] [above]  at (-0,-5.5) {\footnotesize $\lambda=(z+u_0^2/z)/2$};
\draw [->](0.5,-4)--(3.7,-4);
\draw [->](2,-5.2)--(2,-2.5);
\end{tikzpicture}
\end{center}
\caption{ Transformation relation from  $k$ two-sheeted Riemann surface,  $\lambda$-plane and  $z$-plane}
\label{klabdaz}
 \end{figure}
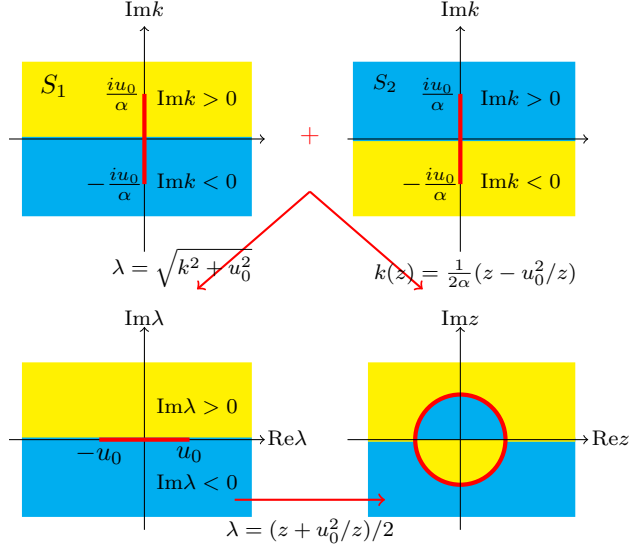

By using the relation
\begin{align*}
{\rm Im}(k(z)\lambda(z))&= {\rm Im}\frac{z^4-u_0^4}{4\alpha z^2}= {\rm Im}\frac{(|z|^4+q_0^4)z^2-2q_0^4(({\rm Re}z)^2-({\rm Im}z)^2)}{4\alpha|z|^4}\\
&= \frac{1}{4\alpha|z|^4}(|z|^4+u_0^4){\rm Im}z^2= \frac{1}{2\alpha|z|^4}(|z|^4+u_0^4){\rm Re}z{\rm Im}z,
\end{align*}
we define two domains $D^+$, $D^-$ and their boundary   $\Sigma$ on $z$-plane by
\begin{align*}
&D^-=\{z:{\rm Re}z{\rm Im}z>0\},\hspace{0.5cm}D^+=\{z:{\rm Re}z{\rm Im}z<0\},\\
&  \Sigma=\{z:{\rm Re}z{\rm Im}z=0\}=  \mathbb{R} \cup i\mathbb{R}\backslash \{0\},
\end{align*}
which   are shown in Figure \ref{fig:figure2}.
\begin{figure}[H]
	\centering
	\includegraphics[width=0.42\linewidth]{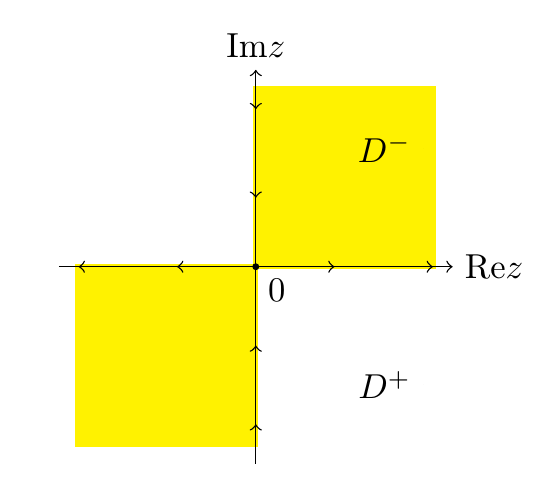}
	\caption{ The domains  $D^-$, $D^+$  and boundary  $\Sigma$.}
	\label{fig:figure2}
\end{figure}

\section{Jost Solutions      }

\quad  From eigenvalues $\pm ik\lambda$, we can get  the   eigenvector matrix of  $X_\pm$ and $T_\pm$ as
\begin{equation}
Y_\pm=\left(\begin{array}{cc}
1 & \frac{u_\pm}{z} \\
-\frac{u_\pm^*}{z} & 1
\end{array}\right) =I+\frac{1}{z}\sigma_3Q_\pm,
\end{equation}
by which   $X_\pm$ and $T_\pm$ are  diagonalized  simultaneously
\begin{equation}
X_\pm=Y_\pm(-ik\lambda\sigma_3)Y_\pm^{-1},\hspace{0.5cm}T_\pm=Y_\pm\left[-ik\lambda\sigma_3(2\alpha k^2-4\alpha-\frac{u_0^2}{\alpha})\right]Y_\pm^{-1}.\label{lax3}
\end{equation}
Direct computation shows  that
\begin{equation}
\det (Y_\pm)=1+\frac{u_0^2}{z^2}\triangleq\gamma,
\end{equation}
and
\begin{equation}
Y_\pm^{-1}=\frac{1}{\gamma}\left(\begin{array}{cc}
1 & -\frac{u_\pm}{z} \\
\frac{u_\pm^*}{z} & 1
\end{array}\right)=\frac{1}{\gamma}(I-\frac{1}{z}\sigma_3Q_\pm).
\end{equation}

Substituting (\ref{lax3}) into (\ref{asymptoticlax}), we immediately obtain
\begin{equation}
(Y_\pm^{-1}\psi)_x=-ik\lambda\sigma_3(Y_\pm^{-1}\psi),\hspace{0.5cm}(Y_\pm^{-1}\psi)_t=-ik\lambda\left(2\alpha k^2-4\alpha-\frac{u_0^2}{\alpha}\right)\sigma_3(Y_\pm^{-1}\psi),
\end{equation}
from which we can derive the solution of the asymptotic spectral problem (\ref{asymptoticlax})
\begin{equation}
\psi_\pm=Y_\pm e^{i\theta(z)\sigma_3},
\end{equation}
where
\begin{equation*}
\theta(z)=-k(z)\lambda(z)[x+(2\alpha k^2(z)-4\alpha-\frac{u_0^2}{\alpha})t].
\end{equation*}
It follows that  the Jost  solutions  $\phi_\pm(x,t,z)$  of the Lax pair (\ref{lax2}) possess the following
asymptotics
\begin{equation}
\phi_\pm \sim Y_\pm e^{i\theta(z)\sigma_3},\hspace{0.5cm}x \rightarrow \pm\infty. \label{asympttu}
\end{equation}

By making transformation
\begin{equation}
\varphi_\pm=\phi_\pm e^{-i\theta(z)\sigma_3},\label{trans2}
\end{equation}
we then  have
\begin{equation*}
\varphi_\pm \sim Y_\pm, \hspace{0.5cm} x \rightarrow \pm\infty.
\end{equation*}
Moreover,  $\varphi_\pm$    satisfy an  equivalent Lax pair
\begin{align}
&(Y_\pm^{-1}\varphi_\pm)_x=ik\lambda[Y_\pm^{-1}\varphi_\pm,\sigma_3]+Y_\pm^{-1}\Delta X_\pm\varphi_\pm, \label{eqlax1}\\
&(Y_\pm^{-1}\varphi_\pm)_t=ik\lambda(2\alpha k^2-4\alpha-\dfrac{1}{\alpha}u_0^2)[Y_\pm^{-1}\varphi_\pm,\sigma_3]+Y_\pm^{-1}\Delta T_\pm\varphi_\pm,\label{eqlax2}
\end{align}
where $\Delta X_\pm=X-X_\pm=ik(Q-Q_\pm)$ and $\Delta T_\pm=T-T_\pm$.  Above  two equations (\ref{eqlax1}) and (\ref{eqlax2}) can be written in full derivative form
\begin{equation}
d(e^{-i\theta(z){\hat{\sigma}}_3}Y_\pm^{-1}\varphi_\pm) = e^{-i\theta(z){\hat{\sigma}}_3}[Y_\pm^{-1}(\Delta X_\pm dx+\Delta T_\pm dt)\varphi_\pm],
\end{equation}
which lead to two  Volterra integral equations
\begin{align}
&\varphi_-(x,t,z)=Y_-+\int_{-\infty}^{x}Y_-e^{-ik\lambda(x-y){\hat{\sigma}}_3}[Y_-^{-1}\Delta X_-\varphi_-(y,t,z)]dy,\label{jost1}\\
&\varphi_+(x,t,z)=Y_+-\int_{x}^{\infty}Y_+e^{-ik\lambda(x-y){\hat{\sigma}}_3}[Y_+^{-1}\Delta X_+\varphi_+(y,t,z)]dy,\label{jost2}
\end{align}
where $z\neq iu_0$.

We   define $\varphi_\pm=(\varphi_{\pm,1},\varphi_{\pm,2})$ with   $\varphi_{\pm,1}$ and $ \varphi_{\pm,2}$ denoting the first and second   column of $\varphi_\pm$
respectively, then the first column of   the equation (\ref{jost1}) can be written as
\begin{equation}
Y_-^{-1}\varphi_{-,1}=\left(\begin{array}{c}
1\\
0
\end{array}\right)+\int_{-\infty}^xG(x-y,z)\Delta X_-\varphi_{-,1}dy,
\end{equation}
where
\begin{equation}
G(x-y,z)=\frac{1}{\gamma}\left(\begin{array}{cc}
1 & -\frac{u_\pm}{z}\\
\frac{u_\pm^*}{z}e^{2ik\lambda(x-y)} & e^{2ik\lambda(x-y)}
\end{array}\right).
\end{equation}
Note that $e^{2ik\lambda(x-y)}=e^{2i(x-y)\text{Re}(k\lambda)}e^{2(x-y)\text{Im}(k\lambda)}$ and   $x-y>0$,   we demonstrate that the first column of $\varphi_-$ is analytical  on $D_-$, denoted by $\varphi_{-,1}^-$.
The same argument shows  that   the second column of $\varphi_-$  is  analytical on  $D_+$, and denoted by  $\varphi_{-,2}^+$.
And    $\varphi_+=(\varphi_{+,1}^+,\varphi_{+,2}^-)$  which  denote the first and second column are analytical on the $D^+$ and $D^-$ respectively.

\section{Scattering Matrix}

\quad  Since ${\rm tr}X={\rm tr}T=0$  in  (\ref{lax2}),   then by using Able formula, we  have
\begin{equation}
(\det\phi_\pm)_x=(\det\phi_\pm)_t=0. \label{abl}
\end{equation}
Again by using the relation
\begin{equation*}
 \det(\varphi_\pm)=\det(\phi_\pm e^{-i\theta(z)\sigma_3})=\det(\phi_\pm),
\end{equation*}
we get  $(\det\varphi_\pm)_x=(\det\varphi_\pm)_t=0$, which means  that  $\det(\varphi_\pm)$ is independent of $x,t$.  So we obtain that
\begin{equation}
\det\varphi_\pm=\lim_{x \to \pm\infty}\det(\varphi_\pm)=\det Y_\pm=\gamma\neq0,  \hspace{0.5cm}z\in D^+\cup D^-,\label{detphi}
\end{equation}
which implies that $\varphi_\pm$ are inverse matrices.

Since   $\phi_\pm$ are two fundamental matrix solutions  of the   Lax  pair (\ref{lax2}),  there exists a linear  relation between $\phi_+$ and $\phi_-$, namely
\begin{equation}
\phi_+(x,t,z)=\phi_-(x,t,z)S(z), \label{scattering}
\end{equation}
where $S(z)$ is called scattering matrix and (\ref{detphi}) implies that $\det S(z)=1$.

 Denoting the scattering matrix by   $S(z)=(s_{ij}(z))_{2\times 2} $,  then   individual columns of  the matrix equation (\ref{scattering})
are
\begin{align}
&\phi_{+,1}=s_{11}(z)\phi_{-,1}+s_{21}(z)\phi_{-,2},\hspace{0.5cm}\phi_{+,2}=s_{12}(z)\phi_{-,1}+s_{22}(z)\phi_{-,2},\label{idvsca1}\\
&\varphi_{+,1}=s_{11}(z)\varphi_{-,1}+s_{21}(z)e^{-2i\theta}\varphi_{-,2},\hspace{0.5cm}\varphi_{+,2}=s_{12}(z)e^{2i\theta}\varphi_{-,1}+s_{22}(z)\varphi_{-,2},\label{idvsca2}
\end{align}
in which  $s_{ij}(z), \ i, j=1,2$ are called scattering data, and   the reflection coefficients are defined by
\begin{equation}
\rho(z)=\frac{s_{21}(z)}{s_{11}(z)},\hspace{0.5cm}\tilde{\rho}(z)=\frac{s_{12}(z)}{s_{22}(z)}.
\end{equation}

Solving above  four  linear systems  (\ref{idvsca1})   and  (\ref{idvsca2}), we  find that
\begin{align}
&s_{11}(z)=\frac{{\rm Wr}(\phi_{+,1},\phi_{-,2})}{\gamma}=\frac{{\rm Wr}(\varphi_{+,1}\varphi_{-,2})}{\gamma},\label{scatteringcoefficient1}\\
&s_{12}(z)=\frac{{\rm Wr}(\phi_{+,2},\phi_{-,2})}{\gamma}=\frac{{\rm Wr}(\varphi_{+,2},\varphi_{-,2})}{e^{2i\theta}\gamma},\\
&s_{21}(z)=\frac{{\rm Wr}(\phi_{-,1},\phi_{+,1})}{\gamma}=\frac{{\rm Wr}(\varphi_{-,1},\varphi_{+,1})}{e^{-2i\theta}\gamma},\\
&s_{22}(z)=\frac{{\rm Wr}(\phi_{-,1},\phi_{+,2})}{\gamma}=\frac{{\rm Wr}(\varphi_{-,1},\varphi_{+,2})}{\gamma},\label{scatteringcoefficient2}
\end{align}
 which together with analyticity of  $\varphi_\pm$   show  that scattering data $s_{11}(z)$ is analytic in $D^+$,  $s_{22}(z)$ is analytic in $D^-$, and $s_{12}(z)$, $s_{21}(z)$ are continuous to $\Sigma$.

\section{Symmetry of $\varphi_\pm$ and $ S(z)$}

\begin{proposition}
	For $z\in D^+$, the Jost solution, scattering matrix and reflection coefficients admit  the following  two kinds of  symmetries
	
	$\bullet$ The first symmetry reduction
	\begin{equation}
	\varphi_\pm(x,t,z)=-\sigma_*\varphi_\pm^*(x,t,z^*)\sigma_*,\hspace{0.5cm}S(z)=-\sigma_*S^*(z^*)\sigma_*,\hspace{0.5cm}\rho(z)=-\tilde{\rho}^*(z^*),
	\label{symmetry1}
	\end{equation}
	\begin{equation}
	\varphi_\pm(x,t,z)=\sigma_1\varphi_\pm^*(x,t,-z^*)\sigma_1,\hspace{0.5cm}S(z)=\sigma_1S^*(-z^*)\sigma_1,\hspace{0.5cm}\rho(z)=\tilde{\rho}^*(-z^*),\label{symmetry2}
	\end{equation}
	where $\sigma_*=\left(\begin{array}{cc}
	0 & 1 \\
	-1 & 0
	\end{array}\right),\hspace{0.5cm}\sigma_1=\left(\begin{array}{cc}
	0 & 1 \\
	1 & 0
	\end{array}\right).$

$\bullet$ The second symmetry reduction
\begin{align}
&\varphi_\pm(x,t,z)=\frac{1}{z}\varphi_\pm\left(x,t,-\frac{u_0^2}{z}\right)\sigma_3Q_\pm\notag,\\
&S(z)=(\sigma_3Q_-)^{-1}S\left(-\frac{u_0^2}{z}\right)\sigma_3Q_+,\hspace{0.5cm}\rho(z)=\frac{u_-}{u_-^*}\tilde{\rho}\left(-\frac{u_0^2}{z}\right).\label{symmetry3}
\end{align}
\end{proposition}

\begin{proof}

We  just need show that   $-\sigma_*\varphi_\pm^*(x,t,z^*)\sigma_*$  is   also  a    solution of (\ref{eqlax1}) and admits  the same   asymptotic behavior
like the Jost solutions  $\varphi_\pm(x,t,z)$.

 By using (\ref{symmetry1}), we  calculate $(-Y_\pm^{-1}\sigma_*\varphi_\pm^*(x,t,z^*)\sigma_*)_x$ and obtain
	\begin{equation}
	\begin{split}
	(-Y_\pm^{-1}\sigma_*\varphi_\pm^*(x,t,z^*)\sigma_*)_x=-Y_\pm^{-1}\sigma_*Y_\pm[-ik(z^*)^*\lambda^*(z^*)[(Y_\pm^{-1})^*&(x,t,z^*)\varphi_\pm(x,t,z^*),\sigma_3]\\+(Y_\pm^{-1})^*(x,t,z^*)(\Delta X_\pm)^*\varphi_\pm^*(x,t,z^*)].
	\end{split}
	\end{equation}
Noticing that
	\begin{align*}
	Y_\pm^{-1}(x,t,z)=-\sigma_*(Y_\pm^{-1})^*(x,t,z^*)\sigma_*,\hspace{0.5cm} \sigma_*\sigma_3\sigma_*=\sigma_3,\hspace{0.5cm}\sigma_*(\Delta X_\pm(z^*))^*\sigma_*=\Delta X_\pm(z),
	\end{align*}
then we have
	\begin{equation}
\begin{split}
(-Y_\pm^{-1}(z)\sigma_*\varphi_\pm^*(x,t,z^*)\sigma_*)_x=-ik(z)\lambda(z)[Y_\pm^{-1}(z)(-\sigma_*\varphi_\pm(x,t,z^*)&\sigma_*),\sigma_3]\\+Y_\pm^{-1}(z)\Delta X_\pm(-\sigma_*\varphi_\pm(x,t,z^*)\sigma_*),
\end{split}
\end{equation}
which implies that  $-\sigma_*\varphi_\pm^*(x,t,z^*)\sigma_*$ is a solution of (\ref{eqlax1}) with  asymptotic behaviors
\begin{equation*}
 -\sigma_*\varphi_\pm^*(x,t,z^*)\sigma_* \sim Y_\pm, \hspace{0.5cm} x \rightarrow \pm\infty.
\end{equation*}
From the uniqueness of the solution   for the   equation (\ref{eqlax1}), we have
\begin{equation}
	\varphi_\pm(x,t,z)=-\sigma_*\varphi_\pm^*(x,t,z^*)\sigma_*.
\end{equation}

Similarly, by using  the symmetry relations
\begin{align*}
&Y_\pm^{-1}(x,t,z)=\sigma_1(Y_\pm^{-1})^*(x,t,-z^*)\sigma_1,\hspace{0.5cm}\sigma_1\sigma_3\sigma_1=-\sigma_3,\hspace{0.5cm}\\
&\theta^*(-z^*)=\theta(z),\hspace{0.5cm}\sigma_1\sigma_3\sigma_1=-\sigma_3,\hspace{0.5cm}
\end{align*}
it is easy to derive the following  symmetry
$$	\varphi_\pm(x,t,z)=\sigma_1\varphi_\pm^*(x,t,-z^*)\sigma_1,\hspace{0.5cm}S(z)=\sigma_1S^*(-z^*)\sigma_1,\hspace{0.5cm}\rho(z)=\tilde{\rho}^*(-z^*).$$

Next, we prove  the second kind of   symmetries (\ref{symmetry3}). From the relations
\begin{equation*}
k\left(-\frac{u_0^2}{z}\right)=k(z),\hspace{0.5cm}\lambda\left(-\frac{u_0^2}{z}\right)=\lambda(z), \ \ \theta\left(-\frac{u_0^2}{z}\right)=-\theta(z),
\end{equation*}
we know that if    $\phi(x,t,z)$ is a solution of the scattering problem (\ref{lax2}),
  then   $\phi(x,t,-\frac{q_0^2}{z})C$  is  also the solution of (\ref{lax2}), where $C$ is  a   determined  $2\times2$ matrix  independent of  $x$  and $t$.
To obtain symmetric relation $\phi(x,t,z)=\phi(x,t,-\frac{q_0^2}{z})C$, we   require $\phi(x,t,-\frac{q_0^2}{z})C$   has the same asymptotic condition (\ref{asympttu})  with  $\phi(x,t,z)$, that is,
\begin{equation}
\phi_\pm\left(x,t,-\frac{u_0^2}{z}\right)C\sim Y_\pm\left(-\frac{u_0^2}{z}\right)e^{-i\theta(z)\sigma_3}C= Y_\pm(z)e^{i\theta\sigma_3}, \ x\rightarrow \pm\infty,
\end{equation}
from which we  find that $C=\frac{1}{z}\sigma_3Q_\pm$.  By uniquenes of the solution of  the spectral  problem (\ref{lax2}), we get
 \begin{equation*}
 \phi(x,t,z)=\frac{1}{z}\phi\left(x,t,-\frac{u_0^2}{z}\right)\sigma_3Q_\pm.
 \end{equation*}
Again by using    (\ref{trans2}),  we have
 \begin{equation}
  \varphi_\pm(x,t,z)=\frac{1}{z}\varphi_\pm\left(x,t,-\frac{u_0^2}{z}\right)\sigma_3Q_\pm.
 \end{equation}

 What will come next are the symmetries of scattering matrix.
 For the individual columns, the above symmetries come to
 \begin{align}
 &\varphi_{\pm,1}(x,t,z)=\sigma_*\varphi_{\pm,2}^*(x,t,z^*), \hspace{0.5cm}\varphi_{\pm,2}(x,t,z)=-\sigma_*\varphi_{\pm,1}^*(x,t,z^*),\label{jostsym1}\\
&\varphi_{\pm,1}(x,t,z)=\sigma_1\varphi_{\pm,2}^*(x,t,-z^*), \hspace{0.5cm}\varphi_{\pm,2}(x,t,z)=\sigma_1\varphi_{\pm,1}^*(x,t,-z^*),\label{jostsym2}\\
 &\varphi_{\pm,1}(x,t,z)=\frac{u_\pm^*}{z}q_\pm^*\varphi_{\pm,2}\left(-\frac{u_0^2}{z}\right),\hspace{0.5cm}\varphi_{\pm,2}(x,t,z)=\frac{u_\pm}{z}q_\pm\varphi_{\pm,1}\left(-\frac{u_0^2}{z}\right)
 \label{jostsym3}.
 \end{align}
 Combining (\ref{jostsym1}) and (\ref{jostsym2}) we can get
\begin{equation}
\varphi_{\pm,1}(x,t,z)=\sigma_3\varphi_{\pm,1}(x,t,-z),\hspace{0.5cm}\varphi_{\pm,2}(x,t,z)=-\sigma_3\varphi_{\pm,2}(x,t,-z)\label{jostsym4}.
\end{equation}

By using  (\ref{jostsym1}) and     (\ref{scattering}),  we obtain a symmetry of scattering matrix
 \begin{equation}
 S^*(z^*)=-\sigma_*S(z)\sigma_*, \ \ \ S(z)=\sigma_1S^*(-z^*)\sigma_1,\label{scasym1}
 \end{equation}
which gives
 \begin{align}
 &s_{11}(z)=s_{22}^*(z^*),\hspace{0.5cm}s_{12}(z)=-s_{21}^*(z^*), \label{scasym1.1}\\
 &s_{11}(z)=s_{22}^*(-z^*),\hspace{0.5cm}s_{12}(z)= s_{21}^*(-z^*). \label{scasym2}
 \end{align}
  Combining (\ref{scasym1.1}) and (\ref{scasym2}) we obtain that $s_{11}(z)$ and $s_{22}(z)$ are even function, and $s_{12}(z)$ and $s_{21}(z)$ are odd function.

By using  (\ref{jostsym3}) and  (\ref{scattering}),  we obtain another symmetry of the scattering matrix
 \begin{equation}
S(z)=(\sigma_3Q_-)^{-1}S\left(-\frac{u_0^2}{z}\right)\sigma_3Q_+,
 \end{equation}
which leads to
 \begin{align}
 &s_{11}^*(z^*)=\frac{u_+}{u_-}s_{11}\left(-\frac{u_0^2}{z}\right),\hspace{0.5cm}s_{12}^*(z^*)=-\frac{u_+^*}{u_-}s_{12}\left(-\frac{u_0^2}{z}\right),\label{elementsym1}\\
 &s_{21}^*(z^*)=-\frac{u_+}{u_-^*}s_{21}\left(-\frac{u_0^2}{z}\right),\hspace{0.5cm}s_{22}^*(z^*)=\frac{u_+^*}{u_-^*}s_{22}\left(-\frac{u_0^2}{z}\right)\label{elementsym2}.
 \end{align}
 Finally, all the above symmetries then give the symmetries for the reflection coefficients
 \begin{equation}
 \rho(z)=\tilde{\rho}^*(-z^*)=-\tilde{\rho}^*(z^*)=\frac{u_-}{u_-^*}\tilde{\rho}\left(-\frac{u_0^2}{z}\right)=-\frac{u_-^*}{u_-}\rho^*\left(-\frac{u_0^2}{z}\right).
 \end{equation}
 So we have done the proof.
\end{proof}
\section{Asymptotics of $\varphi_\pm$ and $S(z)$}

\quad To  get    the Riemann-Hilbert problem in the next section, it is necessary to discuss the asymptotic behaviors of the   Jost
 solutions and scattering matrix as $z \rightarrow \infty$ and $z \rightarrow 0$.
\begin{proposition}
	The Jost solutions posses  the following asymptotic behaviors
	\begin{align}
	&\varphi_\pm(x,t,z)=e^{i\nu_\pm(x,t)\sigma_3}+O(z^{-1}),\hspace{0.5cm}z \rightarrow \infty,\label{asyvarphi1}\\
	&\varphi_\pm(x,t,z)=\frac{1}{z}e^{i\nu_\pm(x,t)\sigma_3}\sigma_3Q_\pm+O(1),\hspace{0.5cm}z \rightarrow 0\label{asyvarphi2},
	\end{align}
	where
	\begin{equation}
	\nu_\pm(x,t)=\frac{1}{2\alpha}\int_{\pm\infty}^x (u_0^2-|u|^2)dy.
	\end{equation}
\end{proposition}
\begin{proof}
	We consider the following  asymptotic  expansions
	\begin{align}
	&\varphi_\pm(x,t,z)=\varphi_\pm^{(0)}(x,t)+\frac{\varphi_\pm^{(1)}(x,t)}{z}+\frac{\varphi_\pm^{(2)}(x,t)}{z^2}+O(z^{-3}),\hspace{0.5cm}\text{as }z \rightarrow \infty.\label{expansion1}
	\end{align}
	Substituting   (\ref{expansion1})   into the Lax pair (\ref{eqlax1}) leads to
\begin{align}
	&[\varphi_\pm^{(0)},\sigma_3]=0,\\
	&[\varphi_\pm^{(1)}-\sigma_3Q_\pm\varphi_\pm^{(0)},\sigma_3]+2(Q-Q_\pm)\varphi_\pm^{(0)}=0,\label{u0}\\
    &\frac{i}{4\alpha}[\varphi_\pm^{(2)}-\sigma_3Q_\pm\varphi_\pm^{(1)},\sigma_3]
    +\frac{i}{2\alpha}((Q-Q_\pm)\varphi_\pm^{(1)}\nonumber-\sigma_3Q_\pm(Q-Q_\pm)\varphi_\pm^{(0)})\notag\\
    &-\frac{i}{2\alpha} (u_0^2-|u|^2)\sigma_3\varphi_\pm^{(0)}=0,\\
	&(\varphi_\pm^{(0)})_x=\frac{i}{4\alpha}[\varphi_\pm^{(2)}-\sigma_3Q_\pm\varphi_\pm^{(1)},\sigma_3]
+\frac{i}{2\alpha}((Q-Q_\pm)\varphi_\pm^{(1)}-\sigma_3Q_\pm(Q-Q_\pm)\varphi_\pm^{(0)})\notag\\
&\quad\quad\quad=\frac{i}{2\alpha} (u_0^2-|u|^2)\sigma_3\varphi_\pm^{(0)},
\end{align}
from  which  we can know  that $\varphi_\pm^{(0)}$ is a diagonal matrix, and
\begin{equation}
\varphi_\pm^{(0)}=e^{i\nu_\pm(x,t)\sigma_3},\label{pop}
\end{equation}
where
\begin{equation}
\nu_\pm(x,t)=\frac{1}{2\alpha}\int_{\pm\infty}^x (u_0^2-|u|^2)dy. \nonumber
\end{equation}
Therefore, we get the asymptotic behavior of the modified Jost solution
\begin{equation*}
\varphi_\pm(x,t,z)=e^{i\nu_\pm(x,t)\sigma_3}+O(z^{-1}),\hspace{0.5cm}z \rightarrow \infty,
\end{equation*}

Again from (\ref{u0}) and (\ref{pop}),   we find  that
\begin{equation}
u=\lim_{z\to\infty}e^{i\nu_\pm(x,t)\sigma_3}(z\varphi_\pm)_{12}\label{u1}
\end{equation}

In a similar way, substituting  the   expansion
\begin{align}
		&\varphi_\pm(x,t,z)=\frac{\tilde{\varphi}_\pm^{(-1)}(x,t)}{z}+\tilde{\varphi}_\pm^{(0)}(x,t)+\tilde{\varphi}_\pm^{(1)}(x,t)z+O(z^2),\hspace{0.5cm}\text{as }z \rightarrow 0\label{expansion2}
	\end{align}
 into the Lax equation (\ref{eqlax1}).   we obtain that
 \begin{equation*}
 \tilde{\varphi}_\pm^{(-1)}=e^{i\nu_\pm(x,t)\sigma_3}C,
 \end{equation*}
 where $C$ is a constant matrix.  Again from the expansion (\ref{expansion2}),  we have
\begin{equation}
\lim_{x\to\pm\infty} z\tilde{\varphi}_\pm=zY_\pm=z(I+\frac{1}{z}\sigma_3Q_\pm)=\lim_{x\to\pm\infty}(\tilde{\varphi}_\pm^{(-1)}+z\tilde{\varphi}_\pm^{(0)}+\cdot\cdot\cdot).
\end{equation}
Therefore $C=\sigma_3Q_\pm$ and $\tilde{\varphi}_\pm^{(-1)}=e^{i\nu_\pm(x,t)\sigma_3}\sigma_3Q_\pm$. Finally, we get the asymptotic behavior
\begin{equation*}
\varphi_\pm(x,t,z)=\frac{1}{z}e^{i\nu_\pm(x,t)\sigma_3}\sigma_3Q_\pm+O(1),\hspace{0.5cm}z \rightarrow 0.
\end{equation*}
\end{proof}

\begin{proposition}
The  scattering matrices admit   asymptotic behaviors
	\begin{align}
	&S(z)=e^{-i\nu_0\sigma_3}+O(z^{-1}),\hspace{0.5cm}z \rightarrow \infty,\label{asympsca1}\\
	&S(z)={\rm diag}\left(\frac{u_-}{u_+},\frac{u_+}{u_-}\right)e^{i\nu_0\sigma_3}+O(z),\hspace{0.5cm}z \rightarrow 0,\label{asympsca2}
	\end{align}
	where
	\begin{equation}
	\nu_0=\frac{1}{2\alpha}\int_{-\infty}^{+\infty}(u_0^2-|u|^2)dy.\label{nu0}
	\end{equation}
\end{proposition}
\begin{proof}

  By using  (\ref{scattering}) and (\ref{asyvarphi1}), for $z \rightarrow \infty,$  we have
	\begin{align*}
e^{i\theta(z)\hat{\sigma}_3}S(z)&=\varphi_-^{-1}\varphi_+=(e^{-i\nu_-(x,t)\sigma_3}+O(z^{-1}))(e^{i\nu_+(x,t)\sigma_3}+O(z^{-1}))\\
	&=e^{-i(\nu_--\nu_+)\sigma_3}+O(z^{-1}) =e^{-i\nu_0\sigma_3}+O(z^{-1}),
	\end{align*}
	where
	\begin{equation*}
	\nu_0=\frac{1}{2\alpha}\int_{-\infty}^{+\infty}(u_0^2-|u|^2)dy.
	\end{equation*}
	
Similarly,
	\begin{align*}
	s_{11}=\frac{Wr(\phi_{+,1},\phi_{-,2})}{\gamma}&=\det\left(\begin{array}{cc}
	O(1) & \frac{1}{z}u_-e^{i\nu_-}+O(1)\\
	\frac{-1}{z}e^{-i\nu_+}u_+^*+O(1) & O(1)
	\end{array}\right)(\frac{z^2}{u_0^2}-\frac{z^4}{u_0^4}+\cdot\cdot\cdot)\\
	&=\frac{u_-}{u_+}e^{i\nu_0}+O(z), \ \ z \rightarrow 0.
	\end{align*}

\end{proof}

\section{Discrete Spectrum and Residue Conditions}

\quad The discrete spectrum of the scattering problem is the set of all values $z\in \mathbb{C}\setminus\Sigma$ which satisfy the eigenfunctions exist in $L^2(\mathbb{R})$.
We suppose that $s_{11}(z)$ has $N_1$ simple zeros $z_1,...,z_{N_1}$ on $D^+\cap\{z\in\mathbb{C}:{\rm Im}z>0,|z|>u_0\}$, and $N_2$ simple  zeros $w_1,...,w_m$ on the circle $\{z=u_0e^{i\varphi}:\frac{\pi}{2}<\varphi<\pi\}$.   The  symmetries  (\ref{symmetry1})-(\ref{symmetry3}) imply that
\begin{equation}
s_{11}(\pm z_n)=0 \Leftrightarrow s_{22}^*(\pm z_n^*)=0 \Leftrightarrow s_{22}\left(\pm \frac{u_0^2}{z_n}\right)=0 \Leftrightarrow s_{11}\left(\pm \frac{u_0^2}{z_n}\right)=0, \hspace{0.5cm}n=1,...,N_1,\nonumber
\end{equation}
and on the circle
\begin{equation}
s_{11}(\pm w_m)=0\Leftrightarrow s_{22}^*(\pm w_m^*)=0, \hspace{0.5cm}m=1,...,N_2. \nonumber
\end{equation}
 It is convenient to define $\zeta_n=z_n$ and $\zeta_{n+N_1}=-\frac{u_0^2}{z_n^*}$ for $n=1,\cdot\cdot\cdot,N_1$,  $\zeta_m=w_{m-2N_1}$ for $m=2N_1+1,\cdot\cdot\cdot,2N_1+N_2$. Therefore the discrete spectrum is
\begin{equation}
Z=\left\{\pm  \zeta_n,\pm  \  \zeta_n^*\right\}_{n=1}^{2N_1+N_2}, \label{spectrals}
\end{equation}
whose distribution on the $z$-plane   are shown  in Figure 3.
\begin{figure}[H]
	\centering
	\includegraphics[width=0.45\linewidth]{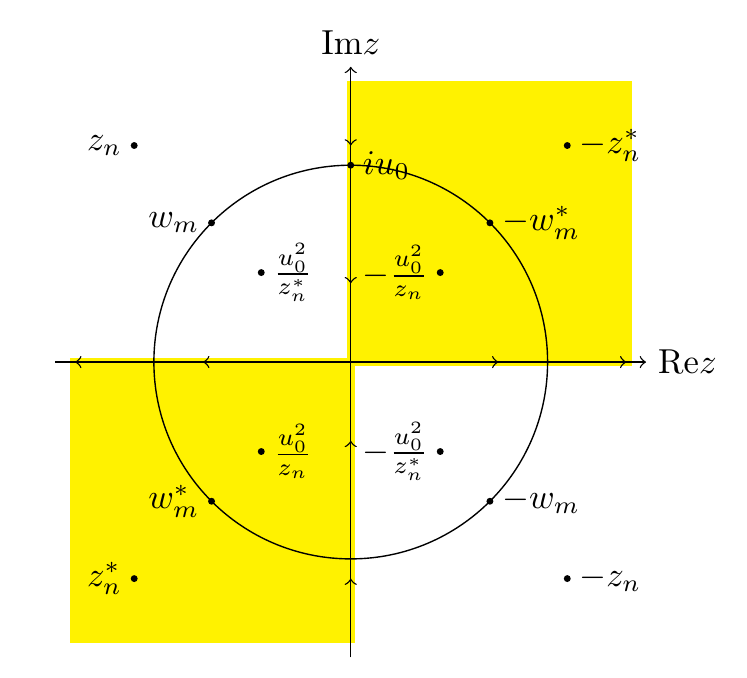}
	\caption{Distribution of the discrete spectrum.}
	\label{fig:figure1}
\end{figure}

All residues on   the  discrete  spectrum are calculated as follows.

\noindent {\bf I.}  For  $s_{11}(\zeta_n)=0, \  n=1,\cdot\cdot\cdot,N_1$.

From determinant   (\ref{scatteringcoefficient1}), we  know  that  the eigenfunction $\varphi_{+,1}(\zeta_n)$ and $\varphi_{-,2}(\zeta_n)$ must
be proportional, then there exist a constant $b_n\neq0$,   such  that
\begin{equation}
\phi_{+,1}(x,t,\zeta_n)=b_n\phi_{-,2}(x,t,\zeta_n), \nonumber
\end{equation}
equivalently,
\begin{equation}
\varphi_{+,1}(x,t,\zeta_n)=b_ne^{2i\theta(\zeta_n)}\varphi_{-,2}(x,t,\zeta_n),\label{tion1}
\end{equation}
by which, the  residue on $z=\zeta_n$  is given by
\begin{equation}
\res_{z=\zeta_n}\left[\frac{\varphi_{+,1}(x,t,z)}{s_{11}(z)}\right]=\lim_{z\to \zeta_n}(z-\zeta_n)\frac{\varphi_{+,1}(x,t,z)}{s_{11}(z)}=C_ne^{-2i\theta(\zeta_n)}\varphi_{-,2}(x,t,\zeta_n),\label{resrelation1}
\end{equation}
 where $C_n=\frac{b_n}{{s_{11}}'(\zeta_n)}$.\\[-1pt]

\noindent {\bf II.}  For   $s_{11}(-\zeta_n)=0, \  n=1,\cdot\cdot\cdot,N_1$.

With the symmetries (\ref{jostsym4}) and (\ref{scatteringcoefficient1}),  we get
\begin{align}
&\varphi_{+,1}(x,t,-\zeta_n)=-b_ne^{-2i\theta(-\zeta_n)}\varphi_{-,2}(x,t,-\zeta_n),\\
&\varphi_{-,2}(x,t,-\zeta_n)=-\sigma_3\varphi_{-,2}(x,t,\zeta_n).
\end{align}
Also noticing that   $(s_{11}(-\zeta_n))'=-s_{11}'(\zeta_n)$ and  $\theta(-\zeta_n)=\theta(\zeta_n)$,   direct computation shows that
\begin{equation}
\res_{z=-\zeta_n}\left[\frac{\varphi_{+,1}(x,t,z)}{s_{11}(z)}\right]=\frac{-b_ne^{-2i\theta(-\zeta_n)}\varphi_{-,2}(x,t,-\zeta_n)}{-s_{11}'(\zeta_n)}=-C_ne^{-2i\theta(\zeta_n)}\sigma_3\varphi_{-,2}(x,t,\zeta_n).\label{resrelation2}
\end{equation}\\[-1pt]

\noindent {\bf III.}  For    $s_{22}(\zeta_n^*)=0, \  n=1,\cdot\cdot\cdot,N_1$.

 There exist  a constant $\tilde{b}_n\neq0 $, such   that
\begin{equation}
\varphi_{+,2}(x,t,\zeta_n^*)=\tilde{b}_ne^{2i\theta(\zeta_n^*)}\varphi_{-,1}(x,t,\zeta_n^*),\label{resrel}
\end{equation}
by which  we can derive that
\begin{equation}
\res_{z=\zeta_n^*}\left[\frac{\varphi_{+,2}(x,t,z)}{s_{22}(z)}\right]=\tilde{C}_ne^{2i\theta(\zeta_n^*)}\varphi_{-,1}(x,t,\zeta_n^*),\label{resrelation2}
\end{equation}
where $\tilde{C}_n=\frac{\tilde{b}_n}{s_{22}'(\zeta_n^*)}$.

Taking the conjugate on both sides of the equation (\ref{resrel})  and multiplying    by $\sigma_*$  leads to
\begin{align}
\sigma_*\varphi_{+,2}^*(x,t,\zeta_n^*)=\tilde{b}_n^*e^{-2i\theta(\zeta_n^*)^*}\sigma_*\varphi_{-,1}^*(x,t,\zeta_n^*), \nonumber
\end{align}
Again  by  the symmetry (\ref{symmetry1}), we  get
\begin{equation}
\varphi_{+,1}(x,t,\zeta_n)=-\tilde{b}_n^*e^{2i\theta(\zeta_n)}\varphi_{-,2}(x,t,\zeta_n).
\end{equation}
 which comparing with (\ref{tion1}) gives  $b_n=-\tilde{b}_n^*$.
 And from the (\ref{scasym1.1}) shows that  $s_{11}'(z)=s_{22}^*(z^*)'$,  so we have $\tilde{C}^*_n=-C_n$.\\[-1pt]

\noindent {\bf IV. }  For $s_{22}(-\zeta_n^*)=0, \  n=1,\cdot\cdot\cdot,N_1$.

  From   (\ref{jostsym4}) and (\ref{scatteringcoefficient2}), we obtain
\begin{align}
&\varphi_{-,1}(x,t,-\zeta_n^*)=\sigma_3\varphi_{-,1}(x,t,\zeta_n^*),\nonumber\\
&\varphi_{+,2}(x,t,-\zeta_n^*)=-\tilde{b}_ne^{2i\theta(-\zeta_n^*)}\varphi_{-,1}(x,t,-\zeta_n^*),\nonumber
\end{align}
by which  the residue is  given  way
\begin{equation}
\res_{z=-\zeta_n^*}\left[\frac{\varphi_{+,2}(x,t,z)}{s_{22}(z)}\right]=\tilde{C}_ne^{2i\theta(\zeta_n^*)}\sigma_3\varphi_{-,1}(x,t,\zeta_n^*).\label{resrelation4}
\end{equation}\\[-1pt]

\noindent {\bf  V. }  For $s_{11}\left(\pm \frac{u_0^2}{\zeta_n^*}\right)=0, \ s_{22}\left(\pm\frac{u_0^2}{\zeta_n}\right)=0, \  n=1,\cdot\cdot\cdot,N_1$.

From  (\ref{elementsym1}),  we get the relation
\begin{equation}
s_{11}\left(-\frac{u_0^2}{\zeta_n^*}\right)=\frac{u_-}{u_+}s_{11}^*(\zeta_n),\nonumber
\end{equation}
which implies that
\begin{equation}
s_{11}'\left(\pm \frac{u_0^2}{\zeta_n}\right)=\left(\frac{\zeta_n^*}{u_0}\right)^2\frac{u_-}{u_+}(s_{11}^*(\zeta_n))'.
\end{equation}
Similarly,
\begin{align}
&s_{22}'\left(-\frac{u_0^2}{\zeta_n}\right)=\left(\frac{\zeta_n}{u_0}\right)^2\frac{u^*_-}{u^*_+}(s_{22}^*(\zeta_n^*))',\nonumber\\
&s_{11}'\left(\frac{u_0^2}{\zeta_n^*}\right)=-\left(\frac{\zeta_n^*}{u_0}\right)^2\frac{u_-}{u_+}(s_{11}^*(\zeta_n))',\nonumber\\
&s_{22}'\left(\frac{u_0^2}{\zeta_n}\right)=-\left(\frac{\zeta_n}{u_0}\right)^2\frac{u^*_-}{u^*_+}(s_{22}^*(\zeta_n^*))'.\nonumber
\end{align}

Finally, combining the above relations, we get
\begin{align}
&\res_{z=-\frac{u_0^2}{\zeta_n^*}}\left[\frac{\varphi_{+,1}(x,t,z)}{s_{11}(z)}\right]=C_{N_1+n}e^{-2i\theta\left(-\frac{u_0^2}{\zeta_n^*}\right)}\varphi_{-,2}\left(x,t,-\frac{u_0^2}{\zeta_n^*}\right),\label{resrelation5}\\
&\res_{z=-\frac{u_0^2}{\zeta_n}}\left[\frac{\varphi_{+,2}(x,t,z)}{s_{22}(z)}\right]=\tilde{C}_{N_1+n}e^{2i\theta\left(-\frac{u_0^2}{\zeta_n}\right)}\varphi_{-,1}\left(x,t,-\frac{u_0^2}{\zeta_n}\right),\label{resrelation6}\\
&\res_{z=\frac{u_0^2}{\zeta_n^*}}\left[\frac{\varphi_{+,1}(x,t,z)}{s_{11}(z)}\right]=-C_{N_1+n}e^{-2i\theta\left(-\frac{u_0^2}{\zeta_n^*}\right)}\sigma_3\varphi_{-,2}\left(x,t,-\frac{u_0^2}{\zeta_n^*}\right),\label{resrelation7}\\
&\res_{z=\frac{u_0^2}{\zeta_n}}\left[\frac{\varphi_{+,2}(x,t,z)}{s_{22}(z)}\right]=\tilde{C}_{N_1+n}e^{2i\theta\left(-\frac{u_0^2}{\zeta_n}\right)}\sigma_3\varphi_{-,1}\left(x,t,-\frac{u_0^2}{\zeta_n}\right),\label{resrelation8}
\end{align}
where
\begin{equation}
C_{N_1+n}=\frac{u_-^*}{u_-}\left(\frac{u_0}{\zeta_n^*}\right)^2\tilde{C}_n,\hspace{0.5cm}\tilde{C}_{N_1+n}=\frac{u_-}{u_-^*}\left(\frac{u_0}{\zeta_n}\right)^2C_n,
\end{equation}
with the relation
\begin{equation}
\tilde{C}_{N_1+n}=-C_{N_1+n}^*.\nonumber
\end{equation}\\[-1pt]

\noindent {\bf  VI. }  For   $s_{11} (\pm w_m )=0$, \ $s_{22} (\pm w_m^* )=0$, \  $n=1,\cdots,N_2$.

Analogously,  we consider the residue conditions at $\pm w_m$ and $\pm w_m^*$ and  get
\begin{align}
&\label{resrelation9}\res_{z=w_m}\frac{\varphi_{+,1}(x,t,z)}{s_{11}(z)}=C_{2N_1+m}e^{-2i\theta(w_m)}\varphi_{-,2}(x,t,w_m),\\
&\label{resrelation10}\res_{z=-w_m}\frac{\varphi_{+,1}(x,t,z)}{s_{11}(z)}=-C_{2N_1+m}e^{-2i\theta(w_m)}\sigma_3\varphi_{-,2}(x,t,w_m),\\
&\label{resrelation11}\res_{z=w_m^*}\frac{\varphi_{+,2}(x,t,z)}{s_{22}(z)}=\tilde{C}_{2N_1+m}e^{2i\theta(w_m^*)}\varphi_{-,1}(x,t,w_m^*),\\
&\label{resrelation12}\res_{z=-w_m^*}\frac{\varphi_{+,2}(x,t,z)}{s_{22}(z)}=\tilde{C}_{2N_1+m}e^{2i\theta(w_m^*)}\sigma_3\varphi_{-,1}(x,t,w_m^*),
\end{align}
where $C_{2N_1+m}=\frac{b_{2N_1+m}}{s_{11}'(w_m)}$, $\tilde{C}_{2N_1+m}=-C_{2N_1+m}^*$ and $b_{2N_1+m}$ are  arbitrary  constants.

\section{Riemann-Hilbert Problem}

\indent Based on above  results analyticity, symmetry and asymptotic of Jost solutions $\varphi_\pm$
and scattering data $s_{ij}(z)$,  we now can establish a generalized Riemann-Hilbert Problem associated with
the solution of the modified NLS equation with nonzero boundary condition.

 \subsection{The   Riemann-Hilbert Problem}

\begin{proposition}
	Define the sectionally meromorphic matrices
	\begin{equation}
	M(x,t,z)=\Bigg\{\begin{array}{ll}
	M^+=\left(\begin{array}{cc}
	\dfrac{\varphi_{+,1}}{s_{11}}, & \varphi_{-,2}\\
	\end{array}\right), &\text{as } z\in D^+,\\
	M^-=\left(\begin{array}{cc}
	\varphi_{-,1}, & \dfrac{\varphi_{+,2}}{s_{22}}\\
	\end{array}\right), &\text{as }z\in D^-,\\
	\end{array}
	\end{equation}
	then a multiplicative matrix Riemann-Hilbert problem is proposed:
	
	$\bullet$ Analyticity: $M(x,t,z)$ is analytic in $\mathbb{C}\setminus\Sigma$ and has single poles.
	
	$\bullet$ Jump condition
	\begin{equation}
	M^-(x,t,z)=M^+(x,t,z)(I-G(x,t,z)),\hspace{0.5cm}z \in \Sigma,\label{jump}
	\end{equation}
	where
	\begin{equation}
	G(x,t,z)=\left(\begin{array}{cc}
	0 & e^{-2i\theta}\tilde{\rho}(z)\\
	e^{2i\theta}\rho(z) & \rho(z)\tilde{\rho}(z)
	\end{array}\right).
	\end{equation}
	
	$\bullet$ Asymptotic behaviors
	\begin{align}
	&M(x,t,z) \sim e^{i\nu_-\sigma_3}+O(z^{-1}),\hspace{0.5cm}z \rightarrow \infty,\label{asymbehv1}\\
	&M(x,t,z) \sim \frac{1}{z}e^{i\nu_-\sigma_3}\sigma_3Q_-+O(1),\hspace{0.5cm}z \rightarrow 0.\label{asymbehv2}
	\end{align}
\end{proposition}
\begin{proof}
	The analyticity  of $M_\pm$  can be find out from   the analyticity of the    $\varphi_\pm$  and $S(z)$.
From (\ref{idvsca1}), we get
	\begin{align}
	\varphi_{-,2}(x,t,z)&=-\tilde{\rho}(z)\varphi_{-,1}(x,t,z)+\frac{\varphi_{+,2}(x,t,z)}{s_{22}(z)},\\
	\frac{\varphi_{+,1}(x,t,z)}{s_{11}(z)}
	&=(1-\rho(z)\tilde{\rho}(z))\varphi_{-,1}(x,t,z)+\rho(z)\frac{\varphi_{+,2}(x,t,z)}{s_{22}(z)},
	\end{align}
which leads to  the jump condition (\ref{jump}).
	
  With the   asymptotic behaviors of the  Jost solution (\ref{asyvarphi2}) and  scattering matrix  (\ref{asympsca2}), we can derive that
	\begin{equation}
	M^+(x,t,z) \sim \frac{1}{z}e^{i\nu_-\sigma_3}\sigma_3Q_-+O(1), \ \ z \rightarrow 0,
	\end{equation}
	\begin{equation}
	M^-(x,t,z) \sim \frac{1}{z}e^{i\nu_-\sigma_3}\sigma_3Q_-+O(1), \ \ z \rightarrow 0.
	\end{equation}
 Similarly, we can get another asymptotic behavior (\ref{asymbehv1}) immediately.
\end{proof}

\subsection{ Reconstruction Formula}

\quad Based on the results in  section 7,   we can obtain the Residue Conditions on the meromorphic matrices $M^\pm$ as follow:
\begin{align}
&\res_{z=\zeta_n}M^+=\left(\begin{array}{cc}
C_ne^{-2i\theta(\zeta_n)}\varphi_{-,2}(x,t,\zeta_n) & 0
\end{array}\right),\hspace{0.5cm}n=1,\cdot\cdot\cdot,2N_1+N_2,\\
&\res_{z=-\zeta_n}M^+=\left(\begin{array}{cc}
-C_ne^{-2i\theta(\zeta_n)}\sigma_3\varphi_{-,2}(x,t,\zeta_n) & 0
\end{array}\right),\hspace{0.5cm}n=1,\cdot\cdot\cdot,2N_1+N_2,\\
&\res_{z=\zeta_n^*}M^-=\left(\begin{array}{cc}
0 & \tilde{C}_ne^{2i\theta(\zeta_n^*)}\varphi_{-,1}(x,t,\zeta_n^*)
\end{array} \right),\hspace{0.5cm}n=1,\cdot\cdot\cdot,2N_1+N_2,\\
&\res_{z=-\zeta_n^*}M^-=\left(\begin{array}{cc}
0 & \tilde{C}_ne^{2i\theta(\zeta_n^*)}\sigma_3\varphi_{-,1}(x,t,\zeta_n^*)
\end{array} \right),\hspace{0.5cm}n=1,\cdot\cdot\cdot,2N_1+N_2.
\end{align}
Subtracting out the asymptotic behaviors and the pole contributions, we obtain that
\begin{align}
\begin{split}
&M^--e^{i\nu_-\sigma_3}-\frac{1}{z}e^{i\nu_-\sigma_3}\sigma_3Q_--\sum_{n=1}^{2N_1+N_2}\frac{\res_{\zeta_n^*}M^-}{z-\zeta_n^*}-\sum_{n=1}^{2N_1+N_2}\frac{\res_{\zeta_n}M^+}{z-\zeta_n}-\\
&\sum_{n=1}^{2N_1+N_2}\frac{\res_{\zeta_n}M^+}{z+\zeta_n}-\sum_{n=1}^{2N_1+N_2}\frac{\res_{-\zeta_n^*}M^-}{z+\zeta_n^*}\\
&=M^+-e^{i\nu_-\sigma_3}-\frac{1}{z}e^{i\nu_-\sigma_3}\sigma_3Q_--\sum_{n=1}^{2N_1+N_2}\frac{\res_{\zeta_n^*}M^-}{z-\zeta_n^*}-\sum_{n=1}^{2N_1+N_2}\frac{\res_{\zeta_n}M^+}{z-\zeta_n}\\
&\sum_{n=1}^{2N_1+N_2}\frac{\res_{\zeta_n}M^+}{z+\zeta_n}-\sum_{n=1}^{2N_1+N_2}\frac{\res_{-\zeta_n^*}M^-}{z+\zeta_n^*}-M^+G.
\end{split}
\end{align}

Apparently, the left-hand side   is analytic in $D^-$ and is $O(z^{-1})$ as $z \rightarrow \infty$, while the sum of the first five terms of the right-hand side   is analytic in $D^+$ and is $O(z^{-1})$ as $z \rightarrow \infty$.

 Using Plemelj's formula, we obtain
\begin{align}
\begin{split} M(x,t,z)=&e^{i\nu_-\sigma_3}+\frac{1}{z}e^{i\nu_-\sigma_3}\sigma_3Q_-+\sum_{n=1}^{2N_1+N_2}\frac{\res_{\zeta_n}M^+}{z-\zeta_n}+\sum_{n=1}^{2N_1+N_2}\frac{\res_{\zeta_n^*}M^-}{z-\zeta_n^*}+\\
&\sum_{n=1}^{2N_1+N_2}\frac{\res_{\zeta_n}M^+}{z+\zeta_n}+\sum_{n=1}^{2N_1+N_2}\frac{\res_{-\zeta_n^*}M^-}{z+\zeta_n^*}+\frac{1}{2\pi i}\int_\Sigma\frac{M^+(x,t,\zeta)}{\zeta-z}G(x,t,z)d\zeta,\hspace{0.5cm}\label{rhpsolution}
\end{split}
\end{align}
Where $z\in\mathbb{C}\setminus\Sigma.$

In the remaining parts of this section we would like to give the reconstruction formula
from the solution of the Riemann-Hilbert problem.
When $z\rightarrow \infty$, from the previous equation we can get
\begin{align}
\begin{split}
M(x,t,z)&=e^{i\nu_-\sigma_3}+\frac{1}{z}\left(e^{i\nu_-\sigma_3}\sigma_3Q_- +\sum_{n=1}^{2N_1+N_2}(\res_{z=\zeta_n}M^++\res_{z=\zeta_n^*}M^-+\res_{z=-\zeta_n}M^++\res_{z=-\zeta_n^*}M^-)\right.\\
&\left.-\frac{1}{2\pi i}\int_\Sigma M^+(x,t,\zeta)G(x,t,\zeta)d\zeta\right)+O(z^{-2}).
\end{split}\label{RHP}
\end{align}

When z=$\zeta_n$, we can calculate the second column of $M^+$ in (\ref{rhpsolution}). Then we obtain
\begin{align}
\begin{split}
\varphi_{-,2}(x,t,\zeta_n)=\left(\begin{array}{c}
\frac{1}{\zeta_n}u_-e^{i\nu_-}\\
e^{-i\nu_-}
\end{array}\right)+2\sum_{k=1}^{2N_1+N_2}&\frac{\tilde{C}_ke^{2i\theta(\zeta_k^*)}}{\zeta_n^2-\zeta_k^{*2}}Z_1\varphi_{-,1}(x,t,\zeta_k^*)\\
&+\frac{1}{2\pi i}\int_\Sigma\frac{M^+(x,t,\zeta)}{\zeta-\zeta_n}G(x,t,\zeta)d\zeta,
\end{split}\label{varphi12}
\end{align}
where $Z_1=\text{diag}\left(
\zeta_n, \zeta_k^* \right), $  for $n=1,...,2N_1+N_2$.

In the same way, when z=$\zeta_n^*$ we can evaluate the first column of $M^-$ and obtain
\begin{align}
\begin{split}
\varphi_{-,1}(x,t,\zeta_n^*)=\left(\begin{array}{c}
e^{i\nu_-}\\
-\frac{1}{\zeta_n}u_-^*e^{-i\nu_-}
\end{array}\right)+2\sum_{j=1}^{2N_1+N_2}&\frac{C_je^{-2i\theta(\zeta_j)}}{\zeta_n^{*2}-\zeta_j^2}Z_2\varphi_{-,2}(x,t,\zeta_j)\\
&+\frac{1}{2\pi i}\int_\Sigma\frac{M^+(x,t,\zeta)}{\zeta-\zeta_n^*}G(x,t,\zeta)d\zeta,
\end{split}\label{varphi11}
\end{align}
where $Z_2=\text{diag}( \zeta_j, \zeta_n^*),$ for $n=1,...,2N_1+N_2$.

By using (\ref{u1}) and  (\ref{RHP}),   we get the reconstruction formula for the potential
\begin{equation}
u(x,t)=e^{2i\nu_-}u_-+e^{i\nu_-}\left\{2\sum_{n=1}^{2N_1+N_2}\tilde{C}_ne^{2i\theta(\zeta_n^*)}\varphi_{-,11}(\zeta_n^*)-\frac{1}{2\pi i}\int_\Sigma(M^+G)_{12}(x,t,\zeta)d\zeta\right\}.\label{potential}
\end{equation}

\subsection{Trace formula and  theta  condition}

Define two  functions $\beta^\pm(z)$ as follow
\begin{align}
\begin{split}\label{peli}
&\beta^+(z)=s_{11}(z)\prod_{n=1}^{2N_1+N_2}\frac{z^2-\zeta_n^{*2}}{z^2-\zeta_n^2}e^{i\nu_0},\\
&\beta^-(z)=s_{22}(z)\prod_{n=1}^{2N_1+N_2}\frac{z^2-\zeta_n^2}{z^2-\zeta_n^{*2}}e^{-i\nu_0},
\end{split}
\end{align}
 which implies the relation $ \beta^+(z)\beta^-(z)=s_{11}(z)s_{22}(z) $ and $ \beta^\pm(z) \rightarrow 1, \ z \rightarrow \pm \infty$.

 From the analyticity of the scattering matrix, we see that the above functions are analytic and have no zeros in $D^+$ and $D^-$ respectively.

From $\det S(z)=1$,  we obtain that
\begin{equation}
\frac{1}{s_{11}s_{22}}=\frac{s_{11}s_{22}-s_{12}s_{21}}{s_{11}s_{22}}=1-\rho(z)\tilde{\rho}(z)=1+\rho(z)\rho^*(z^*),\nonumber
\end{equation}
so that
\begin{equation}
\beta^+(z)\beta^-(z)=\frac{1}{1+\rho(z)\rho^*(z^*)},\hspace{0.5cm}z\in\Sigma.\nonumber
\end{equation}
Taking logarithms to the above relation and using the Plemelj' formula, we have
\begin{equation}
\log\beta_\pm(z)=\mp\frac{1}{2\pi i}\int_\Sigma\frac{\log[1+\rho(\zeta)\rho^*(\zeta^*)]}{\zeta-z}d\zeta,\hspace{0.5cm}z\in D^\pm.
\end{equation}
Substituting them back into (\ref{peli}) leads to
\begin{equation}
s_{11}(z)=\exp\left[-\frac{1}{2\pi i}\int_\Sigma\frac{\log[1+\rho(\zeta)\rho^*(\zeta^*)]}{\zeta-z}d\zeta\right]\\\prod_{n=1}^{2N_1+N_2}\frac{z^2-\zeta_n^2}{z^2-\zeta_n^{*2}}e^{-i\nu_0},\hspace{0.5cm}z\in D^+,\label{trace}\\
\end{equation}
\begin{equation}
s_{22}(z)=\exp\left[\frac{1}{2\pi i}\int_\Sigma\frac{\log[1+\rho(\zeta)\rho^*(\zeta^*)]}{\zeta-z}d\zeta\right]\prod_{n=1}^{2N_1+N_2}\frac{z^2-\zeta_n^{*2}}{z^2-\zeta_n^2}e^{i\nu_0},\hspace{0.5cm}z\in D^-.
\end{equation}
The above formulas are called trace formulas, which express the analytic scattering coefficient in terms of
the discrete eigenvalues and the reflection coefficient.

Taking the limit $z\rightarrow0$  in (\ref{trace}), and using the asymptotic behavior of the scattering matrix (\ref{asympsca2}) we then obtain the following  theta  condition
\begin{equation}
\arg\left(\frac{u_-}{u_+}\right)+2\nu_0=8\sum_{n=1}^{2N_1+N_2}\arg \zeta_n+\frac{1}{2\pi}\int_\Sigma\frac{\log[1+\rho(\zeta)\rho^*(\zeta^*)]}{\zeta}d\zeta.\label{thetacondition}
\end{equation}

\section{Solving the Riemann-Hilbert problem}

\subsection{The formula for $N$-soliton solutions}

We consider the reflectionless potential with the reflection coefficient $\rho(z)=0$, then (\ref{potential}) becomes
\begin{equation}
u(x,t)=e^{2i\nu_-}u_-+e^{i\nu_-}2\sum_{n=1}^{2N_1+N_2}\tilde{C}_ne^{2i\theta(\zeta_n^*)}\varphi_{-,11}(\zeta_n^*).\label{reconstr}
\end{equation}

Denote
\begin{equation}
c_j(x,t,z)=\frac{C_j}{z^2-\zeta_j^2}e^{-2i\theta(x,t,\zeta_j)},\hspace{0.5cm}j=1,\cdot\cdot\cdot,2N_1+N_2.
\end{equation}
We can obtain that
\begin{equation}
c_j^*(x,t,\zeta_k^*)=\frac{C_j^*}{\zeta_k^2-\zeta_j^{*2}}e^{2i\theta(x,t,\zeta_j^*)},\hspace{0.5cm}j=1,\cdot\cdot\cdot,2N_1+N_2.
\end{equation}
 Then from(\ref{varphi12}) and (\ref{varphi11}),  we can drive
 \begin{align}
 &\varphi_{-,12}(x,t,\zeta_j)=\frac{1}{\zeta_j}u_-e^{i\nu_-}+\sum_{k=1}^{2N_1+N_2}2\zeta_jc_k^*(\zeta_j^*)\varphi_{-,11}(x,t,\zeta_k^*),\hspace{0.5cm}j=1,\cdot\cdot\cdot,2N_1+N_2.\label{varphi-12}\\
 &\varphi_{-,11}(x,t,\zeta_n^*)=e^{i\nu_-}+\sum_{j=1}^{2N_1+N_2}2\zeta_jc_j(\zeta_n^*)\varphi_{-,12}(x,t,\zeta_j),\hspace{0.5cm}n=1,\cdot\cdot\cdot,2N_1+N_2\label{varphi-11}.
 \end{align}

Substituting (\ref{varphi-12}) into (\ref{varphi-11}), we get
\begin{equation}
\varphi_{-,11}(x,t,\zeta_n^*)=e^{i\nu_-}+2e^{i\nu_-}\sum_{j=1}^{2N_1+N_2}u_-c_j(\zeta_n^*)+\sum_{j=1}^{2N_1+N_2}\sum_{k=1}^{2N_1+N_2}4\zeta_j^2c_j(\zeta_n^*)c_k^*(\zeta_j^*)\varphi_{-,11}(x,t,\zeta_k^*), \label{systemmu}\\
\end{equation}
for $n=1,\cdot\cdot\cdot,2N_1+N_2.$

Next,   we would like to write above  system (\ref{systemmu})  as a  matrix form, so let
\begin{equation*}
\mathbf{X}=(X_1,...,X_{2N_1+N_2})^t,\hspace{0.5cm}\mathbf{B}=(B_1,...,B_{2N_1+N_2})^t,
\end{equation*}
where
\begin{equation*}
X_n=\varphi_{-,11}(x,t,\zeta_n^*),\hspace{0.5cm}B_n=1+2u_-\sum_{j=1}^{2N_1+N_2}c_j(\zeta_n^*),\hspace{0.5cm}n=1,...,2N_1+N_2.
\end{equation*}
Again define the $(2N_1+N_2)\times(2N_1+N_2)$ matrix $A=(A_{nk})$, where
\begin{equation*}
A_{nk}=-\sum_{j=1}^{2N_1+N_2}4\zeta_j^2c_j(\zeta_n^*)c_k^*(\zeta_j^*),\hspace{0.5cm}n,k=1,...,2N_1+N_2.
\end{equation*}
Then the linear system (\ref{systemmu})  becomes
$$M\mathbf{X}=e^{i\nu_-}\mathbf{B},$$
where $M=I+A=(\mathbf{M}_1,...,\mathbf{M}_{2N_1+N_2})$. The solution of the system is
\begin{equation}
X_n=e^{i\nu_-}\frac{\det M_n^{\text{ext}}}{\det M},\hspace{0.5cm}n=1,...,2N_1+N_2,
\end{equation}
where
\begin{equation*}
M_n^{\rm ext}=(\mathbf{M}_1,...,\mathbf{M}_{n-1},\mathbf{B},\mathbf{M}_{n+1},...,\mathbf{M}_{2N_1+N_2}).
\end{equation*}
Therefore, substituting   the above $X_1,...,X_{2N_1+N_2}$  back into the reconstruction formula (\ref{reconstr}), we then get  formulae for
$N$-soliton solutions
\begin{equation}
u(x,t)=e^{2i\nu_-}u_-+2e^{2i\nu_-}\frac{\det M^{\text{aug}}}{\det M},\label{refpotential}
\end{equation}
where the $(2N_1+N_2+1)\times(2N_1+N_2+1)$ matrix is given by
\begin{equation*}
M^{\text{aug}}=\left(\begin{array}{cc}
0 & \mathbf{Y}\\
\mathbf{B} & M
\end{array}\right),\hspace{0.5cm}\mathbf{Y}=(Y_1,...,Y_{2N_1+N_2}),
\end{equation*}
and
\begin{equation*}
Y_n=\tilde{C}_ne^{2i\theta(x,t,\zeta_n^*)},\hspace{0.5cm}n=1,...,2N_1+N_2.
\end{equation*}

\subsection{One-Soliton Solutions}

\quad   As an application of  $N$-soliton solution formula (\ref{refpotential}),     we  present two kinds of  the one-soliton solutions of modified NLS equation according to different distribution of the spectrum (\ref{spectrals}). Without loss of generality, we take $u_0=1$, due to the fact  if $u(x,t)$ is a solution of (\ref{MNLS}), then $cu(cx,c^t)$ is also the solution of (\ref{MNLS}).\\[0pt]

\noindent {\bf Case I.  $N_1=1$, $N_2=0$:}

In this case,  one eigenvalue  fall to outside of  circle, and suppose that  $\zeta_1=Ze^{i\gamma}$  with $Z>1$, $\gamma \in (\frac{\pi}{2},\pi)$, then the other points in discrete spectrum are  \begin{align}
&-\zeta_1=-Ze^{i\gamma}, \ \   \zeta_2=-\frac{1}{Z}e^{i\gamma}, \ \   -\zeta_2=\frac{1}{Z}e^{i\gamma},\nonumber \\
& \zeta_1^*=Ze^{-i\gamma}, \ \   -\zeta_1^*=-Ze^{-i\gamma}, \ \  \zeta_2^*=-\frac{1}{Z}e^{-i\gamma}, \ \ -\zeta_2^*=\frac{1}{Z}e^{-i\gamma}.\nonumber
\end{align}
By using  the theta condition (\ref{thetacondition}), we
have
\begin{equation}
\arg\left( u_-/u_+ \right)+2\nu_0=16\gamma,\nonumber
\end{equation}
which allow us to set $u_-=1$  and  $u_+=e^{i(2\nu_0-16\gamma)}$.  And we can also know that $C_1=e^{\xi+i\varphi}$, with $\xi$, $\varphi \in \mathbb{R}$  and $C_2=-\frac{1}{Z^2}e^{\xi+i(2\gamma-\varphi)}$.

Substituting  above   data into  formulae   (\ref{refpotential}),  we get the one-soliton solution
\begin{equation}
u(x,t)=e^{2i\nu_-}+2e^{2i\nu_-}\frac{\det\left(\begin{array}{ccc}
	0 & Y_1 & Y_2\\
	B_1 & 1+A_{11} & A_{12}\\
	B_2 & A_{21} & 1+A_{22}
	\end{array}\right)}{\det\left(\begin{array}{cc}
	1+A_{11} & A_{12}\\
	A_{21} & 1+A_{22}
	\end{array}\right)},
\end{equation}
where
\begin{align*}
&\theta(x,t,\zeta_j)=-\frac{1}{4\alpha}(\zeta_j^2-\frac{1}{\zeta_j^2})\left[x+(\frac{1}{2\alpha}(\zeta_j^2+\frac{1}{\zeta_j^2})-4\alpha-\frac{2}{\alpha})t\right],\hspace{0.5cm}j=1,2,\\
&c_j(x,t,z)=\frac{C_j}{z^2-\zeta_j^2}e^{-2i\theta(x,t,\zeta_j)},\hspace{0.5cm}j=1,2,\\
&B_n=1+2\sum_{j=1}^{2}c_j(\zeta_n^*),\ \ \ Y_n=-C_n^*e^{2i\theta(x,t,\zeta_n^*)},\hspace{0.5cm}n=1,2,\\
&A_{nk}=-\sum_{j=1}^{2}4\zeta_j^2c_j(\zeta_n^*)c_k^*(\zeta_j^*),\hspace{0.5cm}n,k=1,2.
\end{align*}
Specially taking  parameters  $\alpha=1$, $Z=2$, $\gamma=\frac{3\pi}{4}$, $\xi=0$, $\varphi=0$,
  we  draw a  graphics   of wave propagation in perspective view and along  x- and $t$-orientation for
 the solution $|u(x,t)|$ in Figure 4 and Figure 5.
This kind of   solution   wave  exhibits localized breather   wave;   While  propagation of the wave
  along  $x$  and  $t$ orientation  are  locally   oscillatory  in the middle,   but  two tail parts are    stationary and  go  to a nonzero constant,
  unlike soliton with zero boundary condition, where it go  to    zero.

\begin{figure}[H]
	\centering
	\subfigure[]{
		\includegraphics[width=5.6cm]{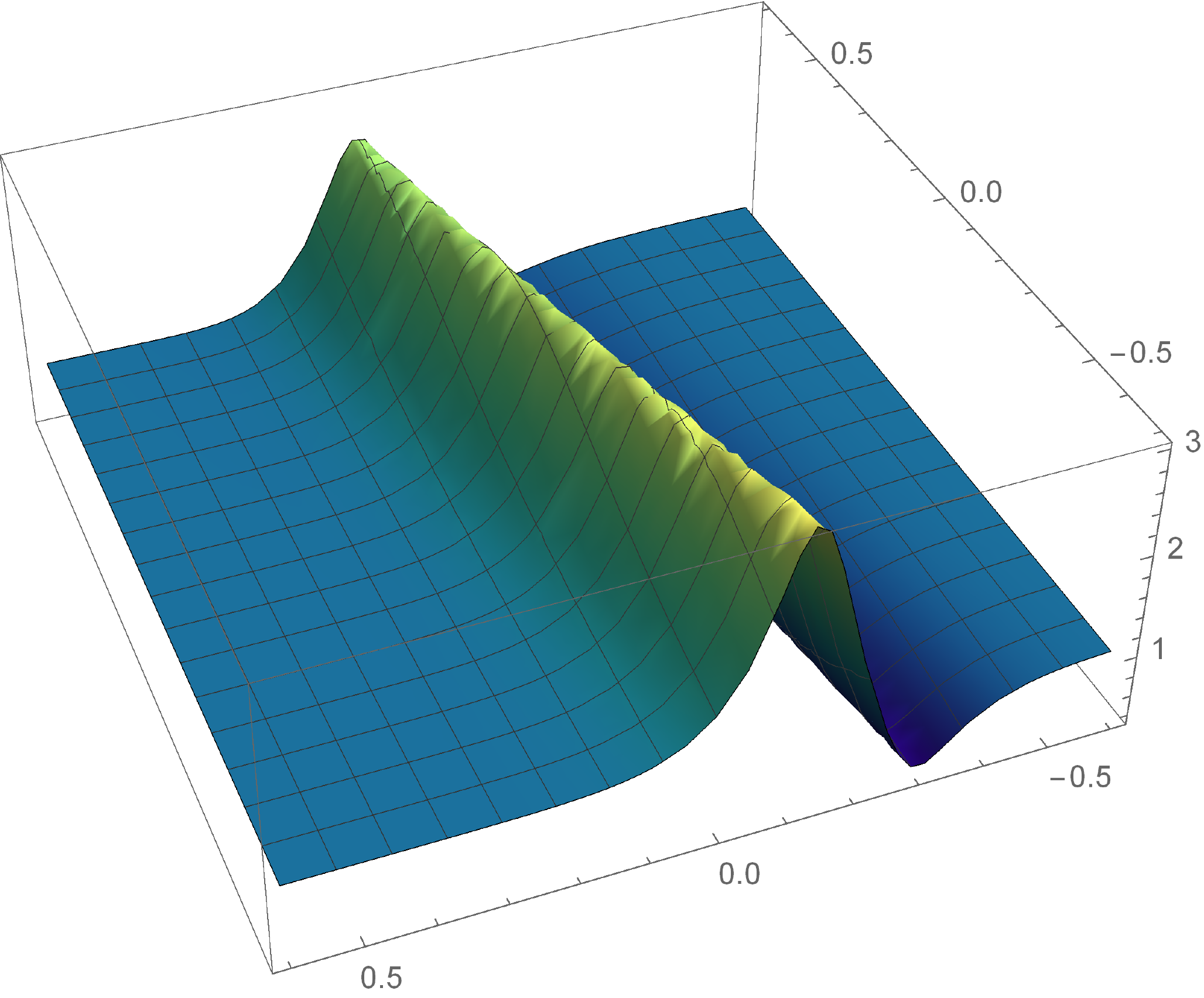}
	}
	\quad
	\subfigure[]{
		\includegraphics[width=0.30\linewidth]{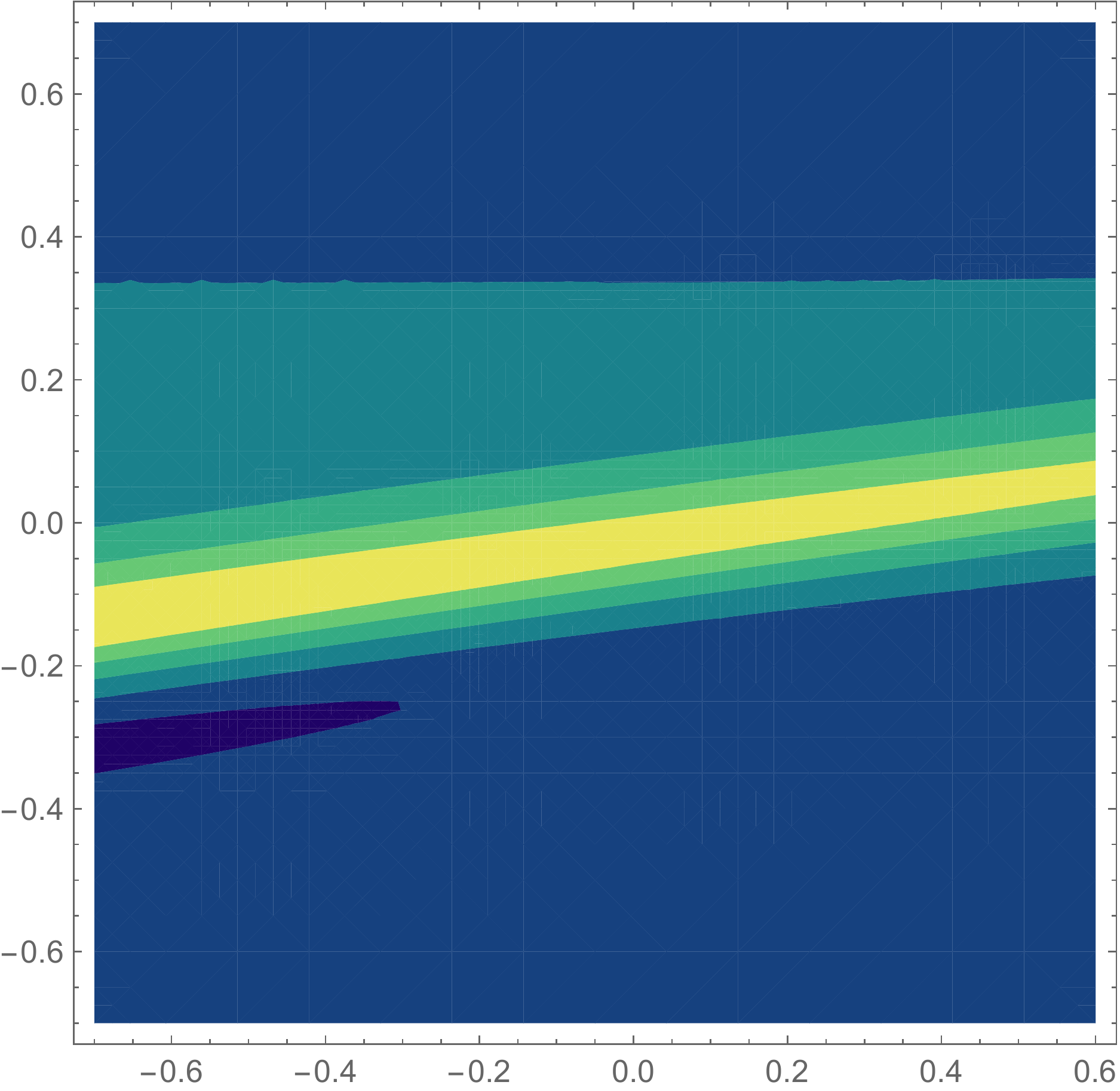}
	}
	\caption{ Breather  solution: (a)  Perspective view of the wave;  (b)  The contour  of the wave.}
\end{figure}
\begin{figure}[H]
	\centering
	\subfigure[]{
	\includegraphics[width=0.41\linewidth]{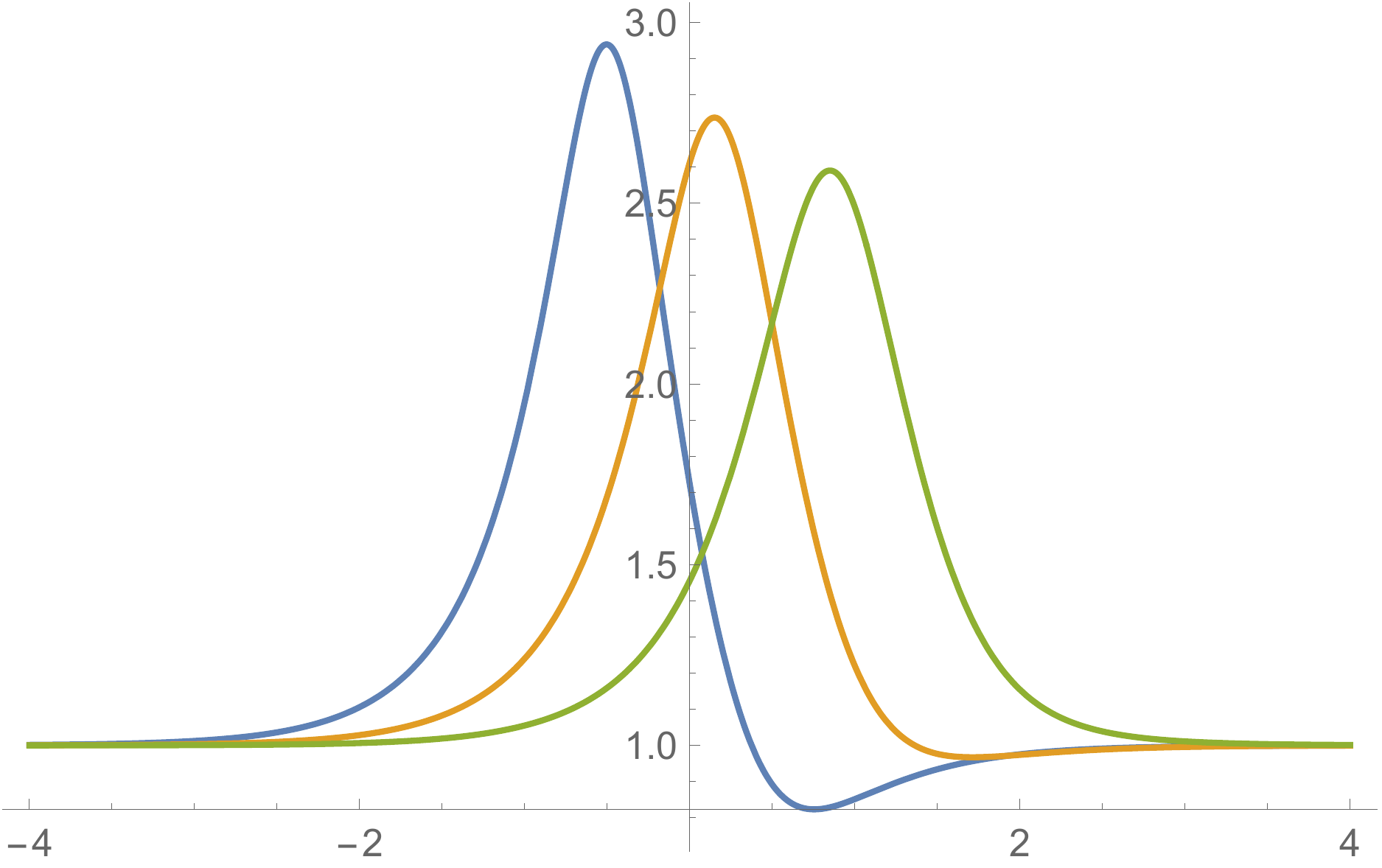}
	}
\subfigure[]{
	\includegraphics[width=0.41\linewidth]{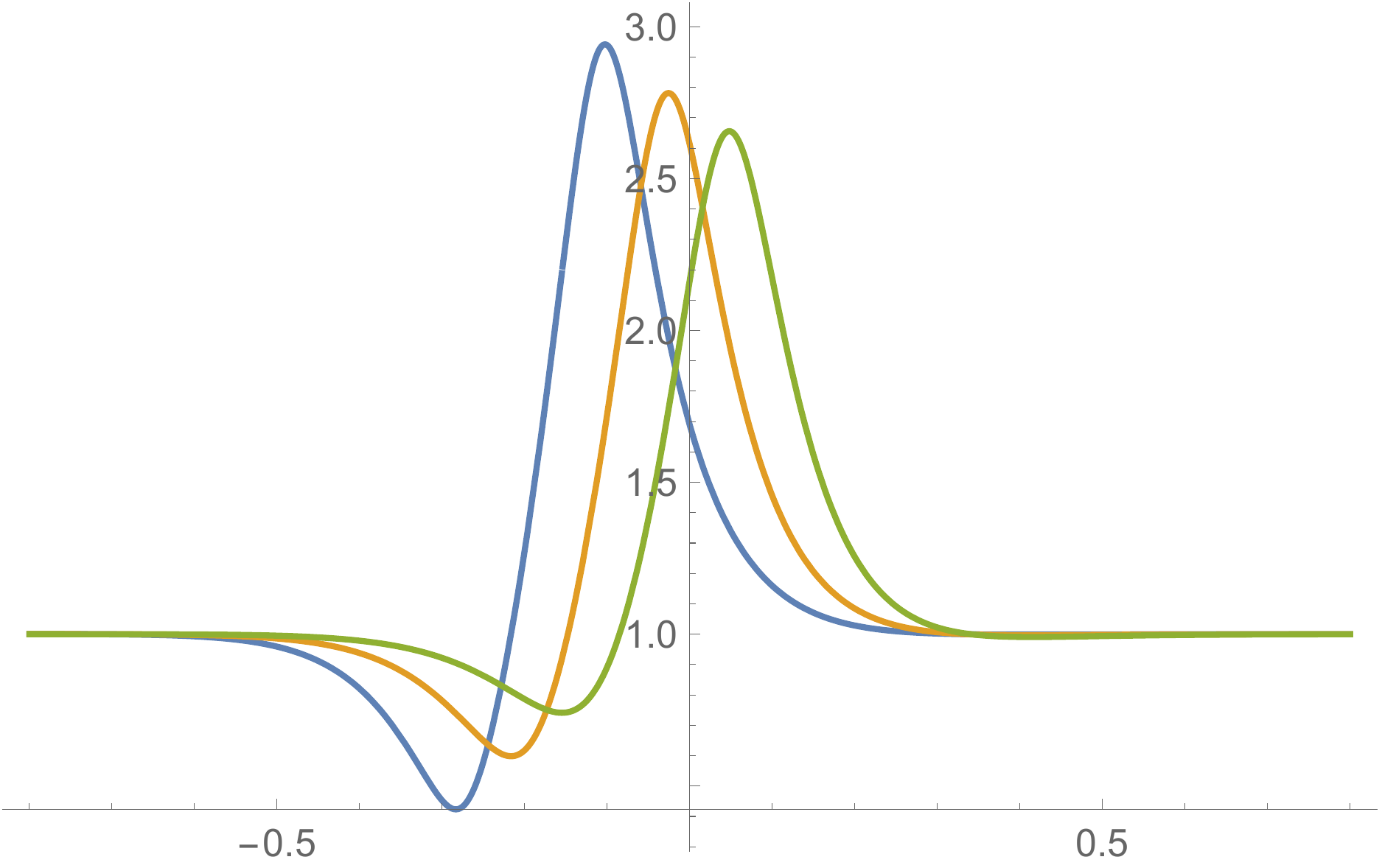}
}
	\caption{Breather  solution:  (a) Propagation pattern  of the wave along $x$-orientation  with $t=-0.1$(blue),  $t=0$(orange),  $t=0.1$(green); (b) Propagation pattern of the  wave along $t$-orientation with $x=-0.5$(blue),  $x=0$(orange),  $x=0.5$(green).}
	\label{fig:3}
\end{figure}
\noindent {\bf Case II.  $N_1=0$, $N_2=1$:}

Taking a eigenvalue $\zeta_1=e^{i\beta}$ with $\beta\in(\frac{\pi}{2},\pi)$ on the circle $|z|=1$, then the discrete spectrum are  $\{e^{i\beta},-e^{i\beta},e^{-i\beta},-e^{-i\beta}\}$.
By using theta condition (\ref{thetacondition}),  we obtain that
\begin{equation}
\arg(u_-/u_+)=8\beta-2\nu_0,
\end{equation}
which allow us to let  $u_-=1$ and  $u_+=e^{i(2\nu_0-8\beta)}$. Take  $C_1=e^{i\tau+\kappa}$  with $\tau,\kappa\in\mathbb{R}$,    we find another kind of  one-soliton solution
\begin{equation}
u(x,t)=e^{2i\nu_-}+2e^{2i\nu_-}\frac{\det\left(\begin{array}{cc}
	0 & Y \\
	B & 1+A
	\end{array}\right)}{1+A},
\end{equation}
where
\begin{align*}
&\theta(x,t,\zeta_1)=-\frac{1}{2\alpha}\sinh(\beta)\left[x+(\frac{1}{\alpha}\cosh(\beta)-4\alpha-\frac{2}{\alpha})t\right],\\
&c_1(x,t,z)=\frac{C_1}{z^2-\zeta_1^2}e^{-2i\theta(x,t,\zeta_1)},\ \ B=1+2c_1(\zeta_n^*),\\
&A=-4\zeta_1^2|c_1(\zeta_1^*)|^2,\ \ Y=-C_1^*e^{2i\theta(x,t,\zeta_n^*)}.
\end{align*}
Specially taking  parameters  $\alpha=1$, $\beta=\frac{3\pi}{4}$, $\tau=\frac{\pi}{2}$, $\kappa=0$,  propagation of the wave in perspective view and along  x- and $t$-orientation for
 the solution $|u(x,t)|$ are  given in Figure 6 and Figure 7.
This kind of  one-soliton  wave  exhibits    bell-type
localized wave like soliton wave zero boundary condition.    While  propagation  of wave
  along  $x$  and  $t$ orientation  are  locally   oscillatory  in the middle,    but  two tail parts  are    stationary and go  to zero.

\begin{figure}[H]
	\centering
	\subfigure[]{
		\includegraphics[width=5.4cm]{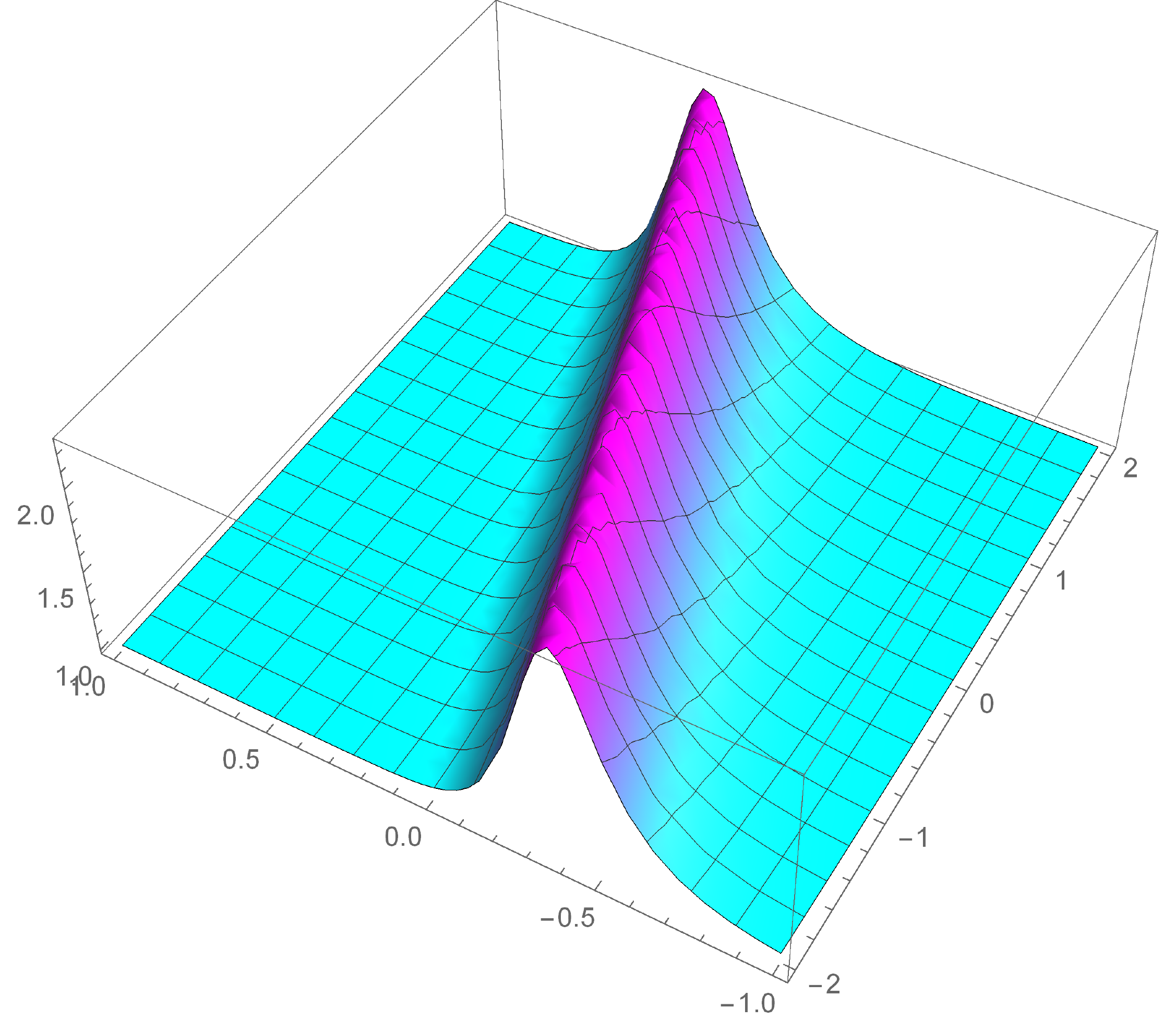}
	}
	\quad
	\subfigure[]{
		\includegraphics[width=0.31\linewidth]{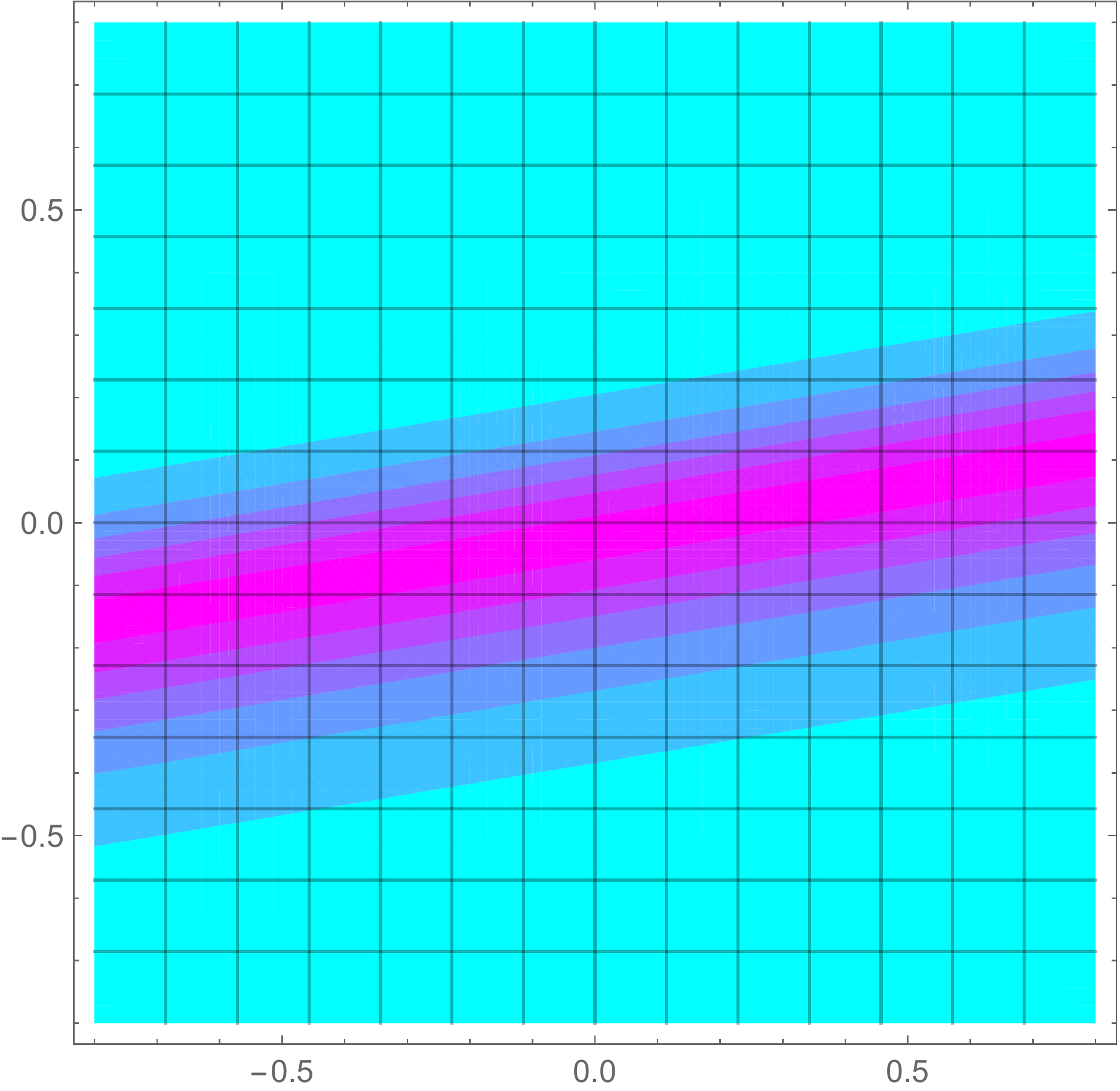}
	}
	\caption{One-solition solution: (a)  Perspective view of the wave;  (b)  The contour  of the wave. }
\end{figure}
\begin{figure}[H]
	\centering
	\subfigure[]{
		\includegraphics[width=0.41\linewidth]{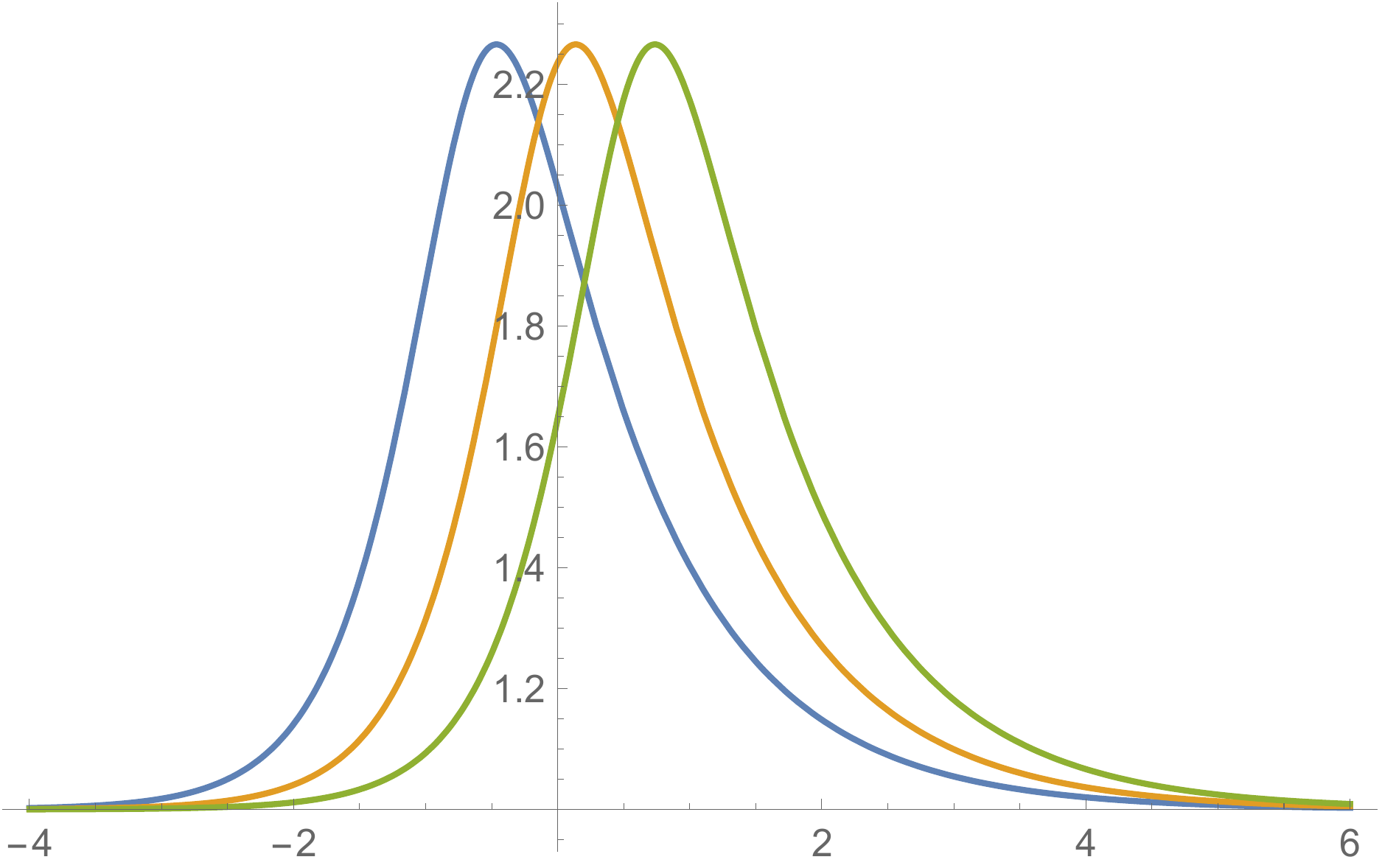}
	}
	\subfigure[]{
		\includegraphics[width=0.41\linewidth]{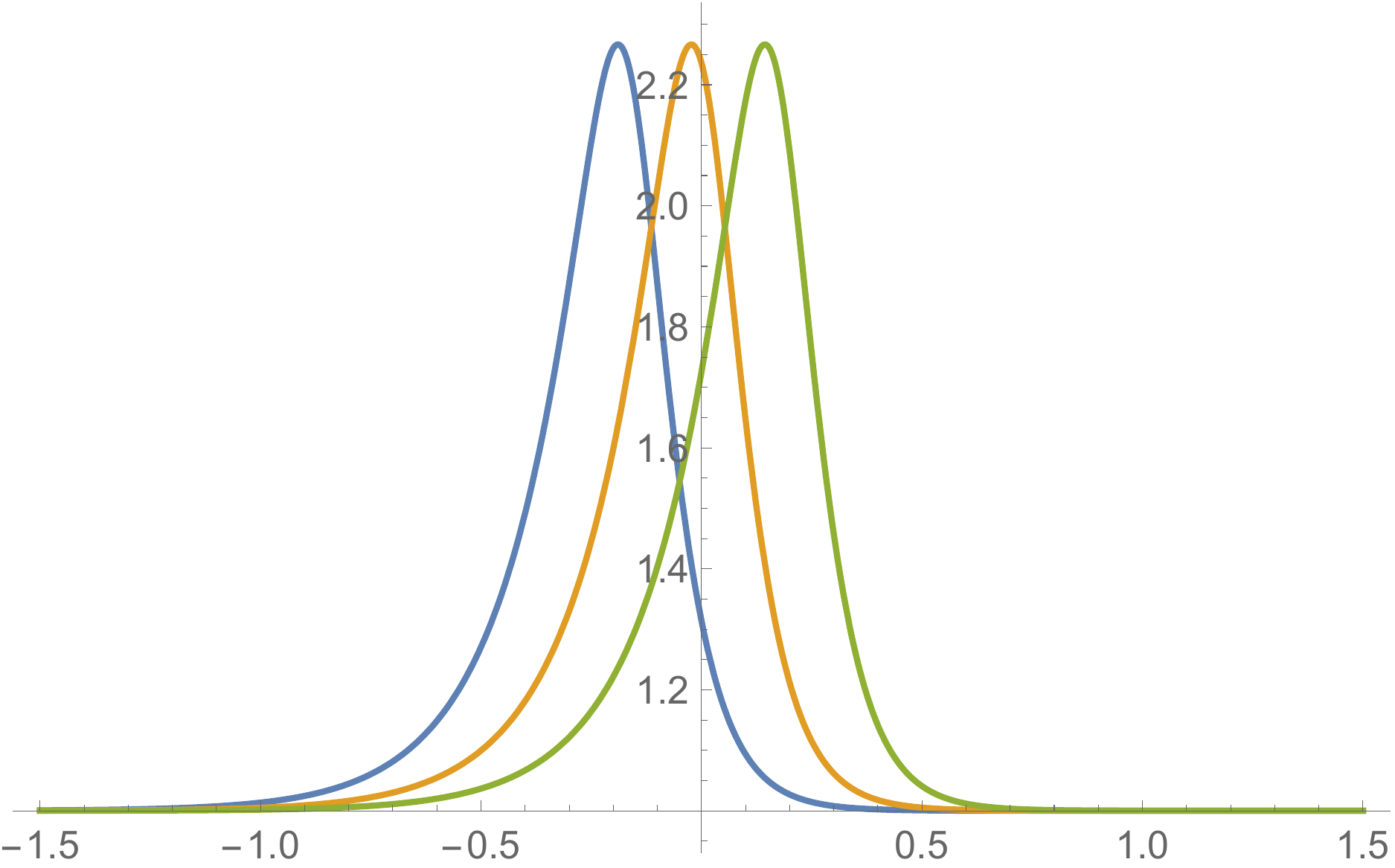}
	}
	\caption{ One-soliton solution:  (a) Propagation pattern of the  wave along $x$-orientation  with $t=-0.1$(blue),  $t=0$(orange),  $t=0.1$(green); (b)
Propagation pattern of the  wave along $t$-orientation with  $x=-1$(blue),  $x=0$(orange),  $x=1$(green).}
	\label{fig:3}
\end{figure}

\subsection{Two-Soliton Solutions}

 Here we  consider   three kinds of two-soliton solutions   according to different distribution of the spectrum
$Z=\left\{\pm  \zeta_n,\pm  \zeta_n^*\right\}_{n=1}^{2N_1+N_2}$.\\

\noindent {\bf Case I.  $N_1=1$, $N_2=1$:}

We take two eigenvalues
 $$\zeta_1=Ze^{i\gamma}, \ \ \zeta_3=e^{i\beta}, \  \  Z>1, \ \  \gamma, \ \beta \in (\frac{\pi}{2},\pi),  $$
 then the other points in discrete spectrum are
\begin{align}
&-\zeta_1=-Ze^{i\gamma}, \ \ \zeta_2=-\frac{1}{Z}e^{i\gamma}, \ \  -\zeta_2=\frac{1}{Z}e^{i\gamma}, \ \  \zeta_1^*=Ze^{-i\gamma},\ \ -\zeta_1^*=-Ze^{-i\gamma},\nonumber\\
  &  \zeta_2^*=-\frac{1}{Z}e^{-i\gamma},  \ \ -\zeta_2^*=\frac{1}{Z}e^{-i\gamma}, \ \ -\zeta_3=-e^{i\beta}, \ \  \zeta_3^*=e^{-i\beta},\ \ -\zeta_3^*=-e^{-i\beta}.\nonumber
  \end{align}

By using  the theta condition (\ref{thetacondition}), we
have
\begin{equation}
\arg\left( u_-/u_+ \right)+2\nu_0=16\gamma+8\beta,\nonumber
\end{equation}
from which we choose  $u_-=1$  and  $u_+=e^{i(2\nu_0-16\gamma-8\beta)}$.  Let $C_1=e^{\xi+i\varphi}$, $C_3=e^{i\tau+\kappa}$, with $\xi$, $\varphi$, $\tau$, $\kappa \in \mathbb{R}$, then $C_2=-\frac{1}{Z^2}e^{\xi+i(2\gamma-\varphi)}$.\\
Substituting  above   data into  formulae   (\ref{refpotential}),  we get one kind of  two-soliton solution
\begin{equation}
u(x,t)=e^{2i\nu_-}+2e^{2i\nu_-}\frac{\det\left(\begin{array}{cccc}
	0 & Y_1 & Y_2 & Y_3\\
	B_1 & 1+A_{11} & A_{12} & A_{13}\\
	B_2 & A_{21} & 1+A_{22} & A_{23}\\
	B_3 & A_{31} & A_{32} & 1+A_{33}
	\end{array}\right)}{\det\left(\begin{array}{ccc}
	1+A_{11} & A_{12} & A_{13}\\
	A_{21} & 1+A_{22} & A_{23}\\
	A_{31} & A_{32} & 1+A_{33}
	\end{array}\right)},
\end{equation}
where
\begin{align*}
&\theta(x,t,\zeta_j)=-\frac{1}{4\alpha}(\zeta_j^2-\frac{1}{\zeta_j^2})\left[x+(\frac{1}{2\alpha}(\zeta_j^2+\frac{1}{\zeta_j^2})-4\alpha-\frac{2}{\alpha})t\right],\hspace{0.5cm}j=1,2,3,\\
&c_j(x,t,z)=\frac{C_j}{z^2-\zeta_j^2}e^{-2i\theta(x,t,\zeta_j)},\hspace{0.5cm}j=1,2,3,\\
&B_n=1+2\sum_{j=1}^{3}c_j(\zeta_n^*),\ \ Y_n=-C_n^*e^{2i\theta(x,t,\zeta_n^*)},\hspace{0.5cm}n=1,2,3.\\
&A_{nk}=-\sum_{j=1}^{3}4\zeta_j^2c_j(\zeta_n^*)c_k^*(\zeta_j^*),\hspace{0.5cm}n,k=1,2,3.
\end{align*}
The propagation  feature    of  this two-soliton  wave  are shown in   Figure 8 and Figure 9.
 This kind of two-soliton wave  exhibits   two-peak  localized wave,   there is  small   oscillatory waves  on one peak,
 and  another peak is smooth.  While    propagation pattern    of   the two-soliton  wave  along  x-orientation   and  $t$-orientation   are locally   oscillatory  in the middle,
     and  two  tail  parts   are    stationary and   go    to  a  nonzero  constant.

\begin{figure}[H]
	\centering
	\subfigure[]{
		\includegraphics[width=5.4cm]{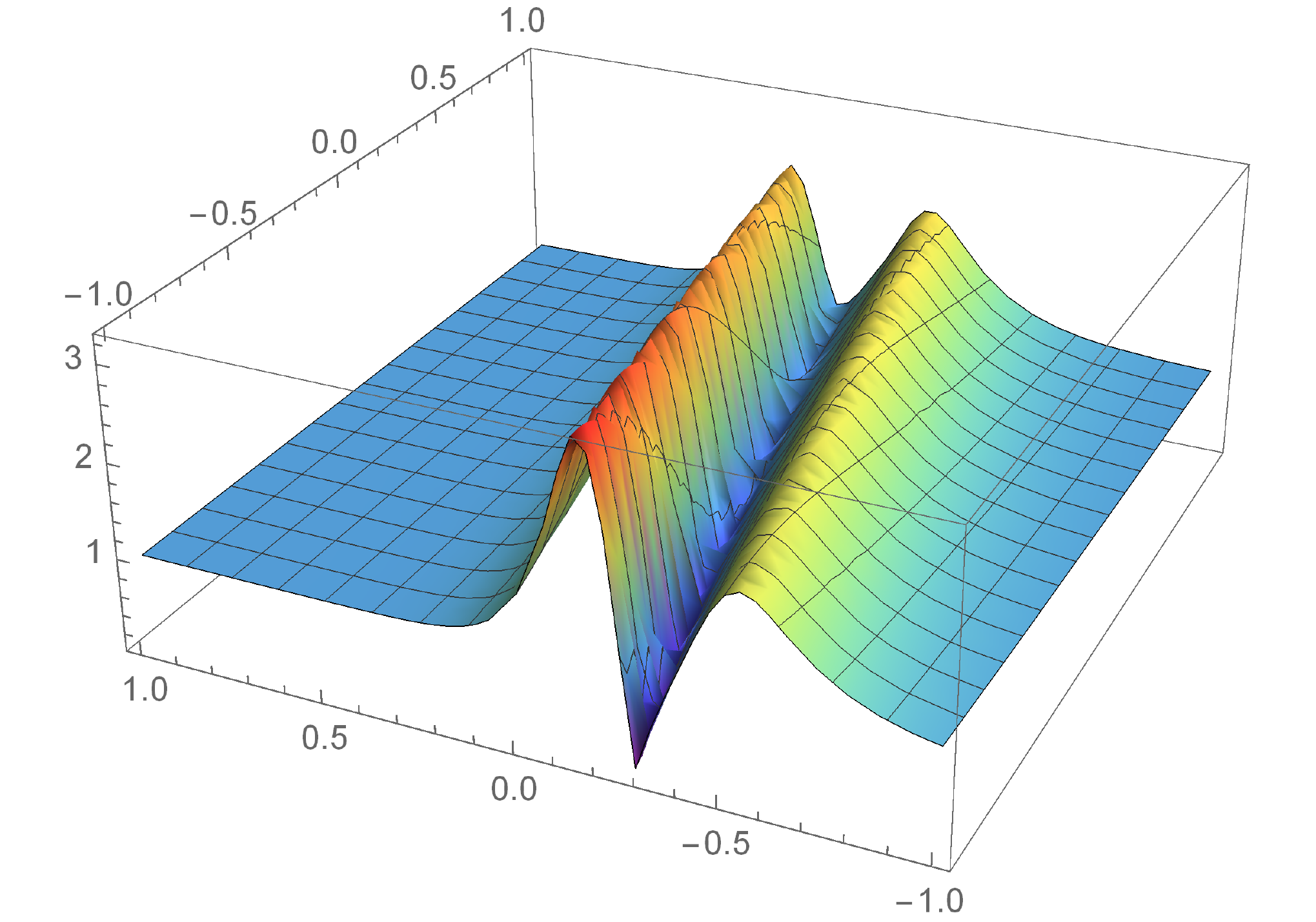}
	}
	\quad
	\subfigure[]{
		\includegraphics[width=0.31\linewidth]{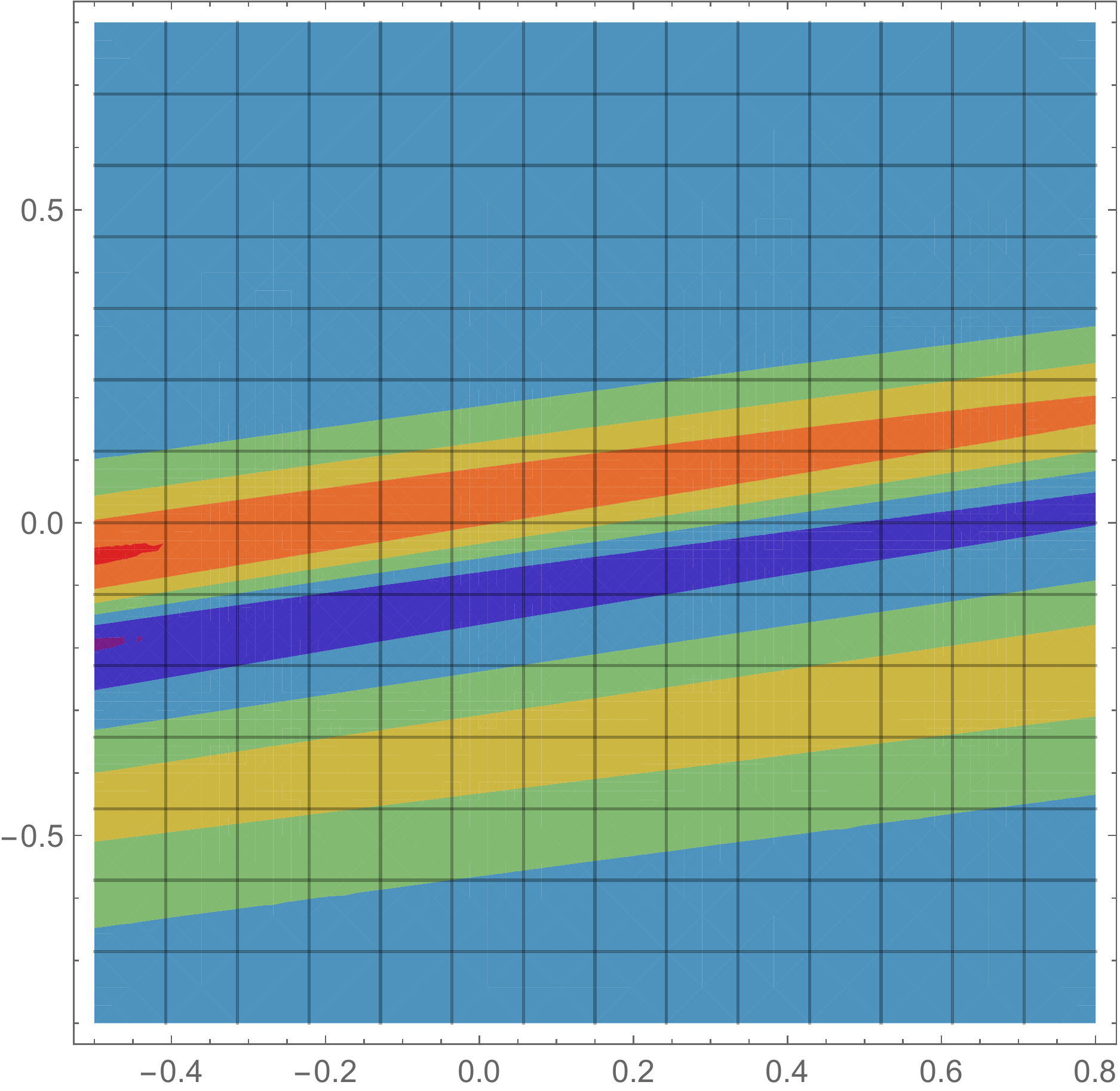}
	}
	\caption{Two-solition solution: (a)  Perspective view of the wave;  (b)  The contour  of the wave.}
\end{figure}
\begin{figure}[H]
	\centering
	\subfigure[]{
		\includegraphics[width=0.41\linewidth]{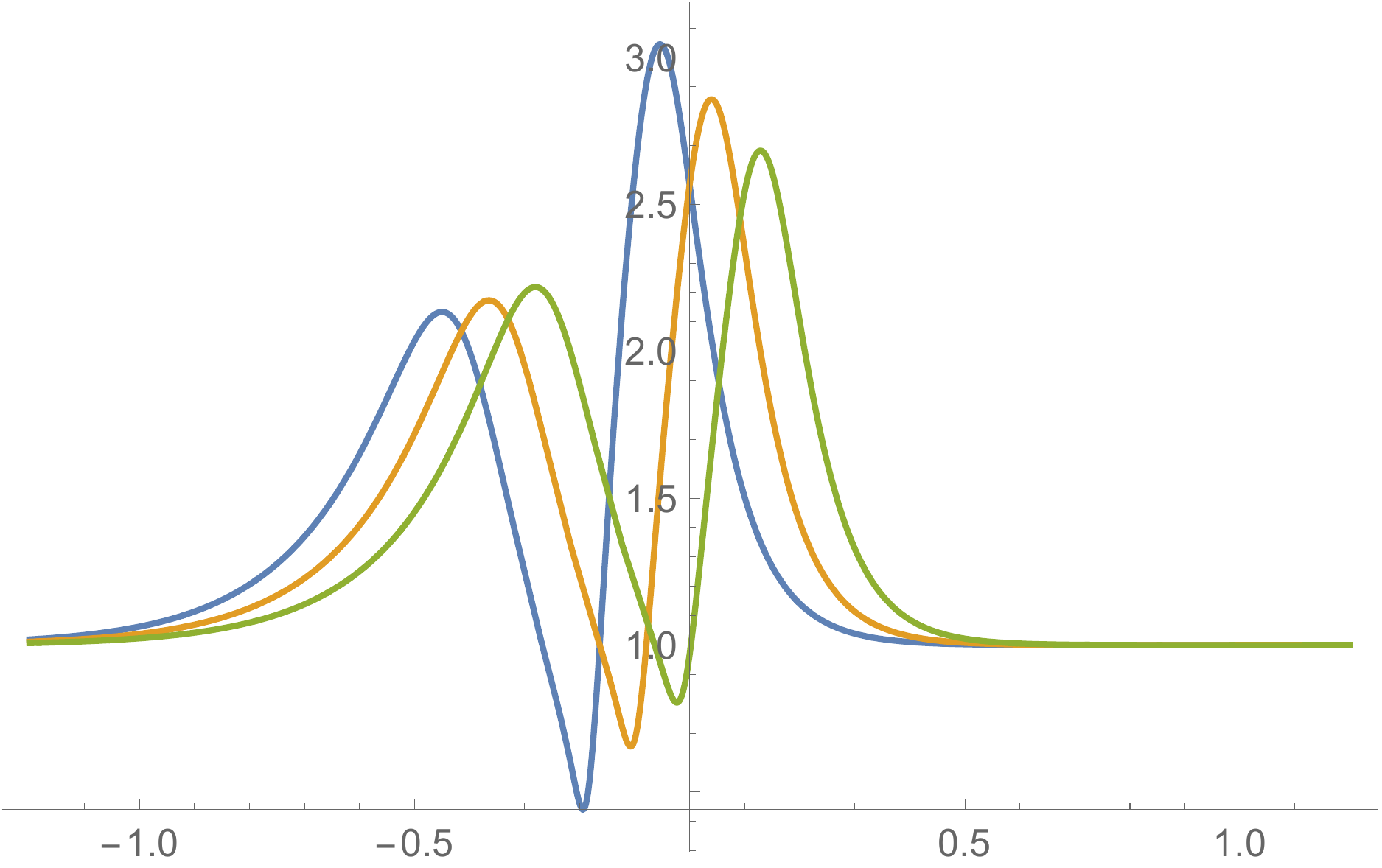}
	}
	\subfigure[]{
		\includegraphics[width=0.41\linewidth]{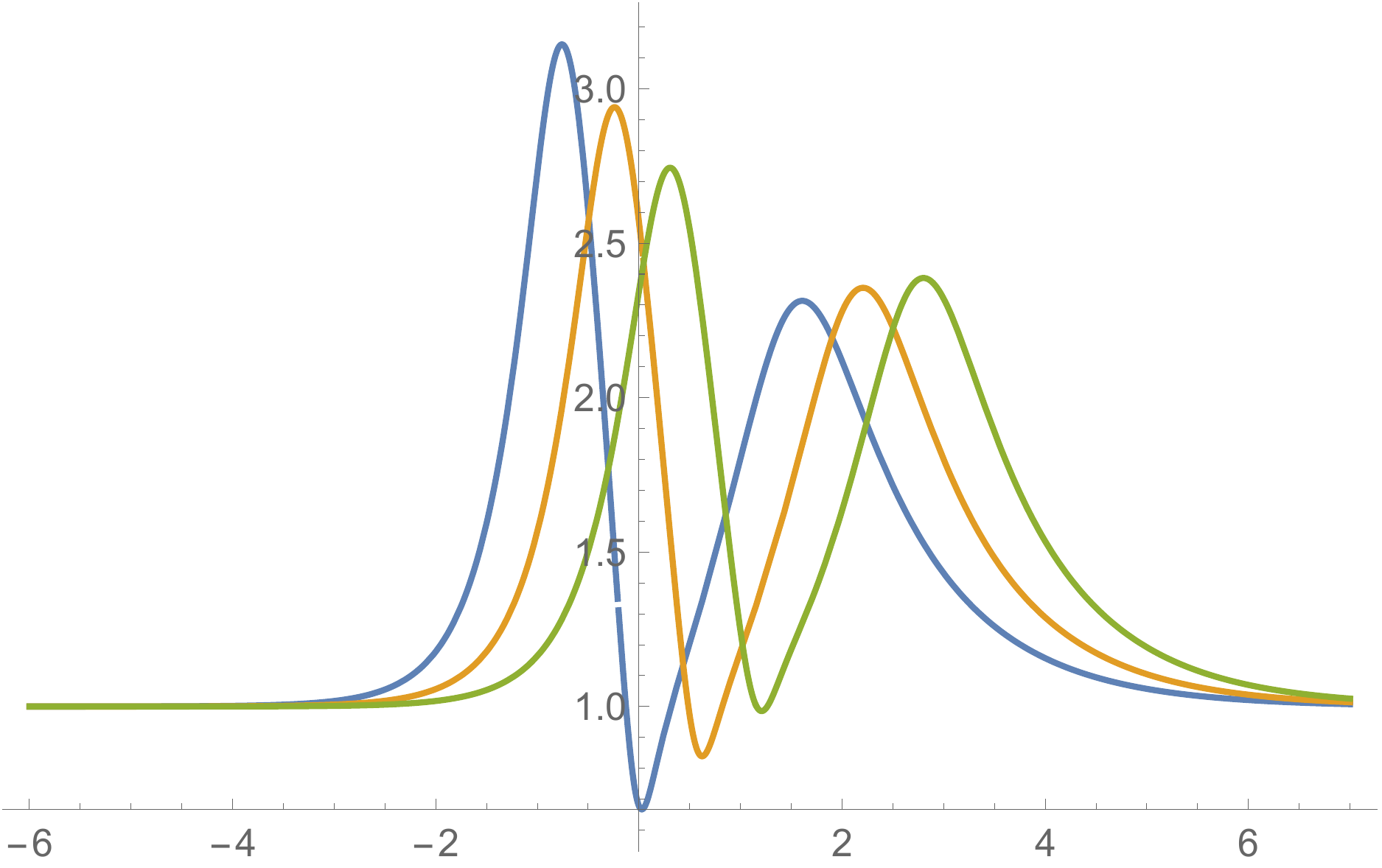}
	}
	\caption{Two-soliton solution: (a) Propagation pattern of the wave along $x$-orientation  with $t=-0.1$ (blue),  $t=0$ (orange),  $t=0.1$ (green); (b) Propagation pattern of the wave along $t$-orientation with $x=-0.5$ (blue),  $x=0$ (orange),  $x=0.5$ (green).}
	\label{fig:3}
\end{figure}

\noindent {\bf Case II.  $N_1=0$, $N_2=2$:}

 Choosing two eigenvalues  $\zeta_j=e^{i\beta_j}$, with $\beta_j\in(\frac{\pi}{2},\pi)$, $j=1,2$,  then    discrete spectrum are
  $\{e^{i\beta_j},-e^{i\beta_j},e^{-i\beta_j},-e^{-i\beta_j}\}_{j=1,2}$. By using theta condition (\ref{thetacondition}), we get
\begin{equation}
\arg(u_-/u_+)=8\beta_1+8\beta_2+-2\nu_0.
\end{equation}
We set $u_-=1$ then $u_+=e^{i(2\nu_0-8\beta_1-8\beta_2)}$. Let $C_j=e^{i\tau_j+\kappa_j}$, with $\tau_j,\kappa_j\in\mathbb{R}$, $j=1,2$. Once again we get the soliton solution with parameters $N_1=0$ and $N_2=2$ by
substituting  above   data into  formulae   (\ref{refpotential}).
\begin{equation}
u(x,t)=e^{2i\nu_-}+2e^{2i\nu_-}\frac{\det\left(\begin{array}{ccc}
	0 & Y_1 & Y_2\\
	B_1 & 1+A_{11} & A_{12}\\
	B_2 & A_{21} & 1+A_{22}
	\end{array}\right)}{\det\left(\begin{array}{cc}
	1+A_{11} & A_{12}\\
	A_{21} & 1+A_{22}
	\end{array}\right)},
\end{equation}
where
\begin{align*}
&\theta(x,t,\zeta_j)=-\frac{1}{4\alpha}(\zeta_j^2-\frac{1}{\zeta_j^2})\left[x+(\frac{1}{2\alpha}(\zeta_j^2+\frac{1}{\zeta_j^2})-4\alpha-\frac{2}{\alpha})t\right],\hspace{0.5cm}j=1,2,\\
&c_j(x,t,z)=\frac{C_j}{z^2-\zeta_j^2}e^{-2i\theta(x,t,\zeta_j)},\hspace{0.5cm}j=1,2,\\
&B_n=1+2\sum_{j=1}^{2}c_j(\zeta_n^*),\ \ Y_n=-C_n^*e^{2i\theta(x,t,\zeta_n^*)},\hspace{0.5cm}n=1,2,\\
&A_{nk}=-\sum_{j=1}^{2}4\zeta_j^2c_j(\zeta_n^*)c_k^*(\zeta_j^*),\hspace{0.5cm}n,k=1,2.
\end{align*}
 The wave propagation  properties  of this two-soliton solution  are shown in   Figure 10 and Figure 11.
 This kind of  two-soliton wave   exhibits    two-peak  localized wave,   there  are no  oscillatory waves  and  smooth   on  both  peaks;
   While propagation of the    wave   along  $x$-orientation   and  $t$-orientation   are locally   oscillatory  in the middle,
     and  two  tail  parts   are    stationary  and   go  to zero.

\begin{figure}[H]
	\centering
	\subfigure[]{
		\includegraphics[width=5.4cm]{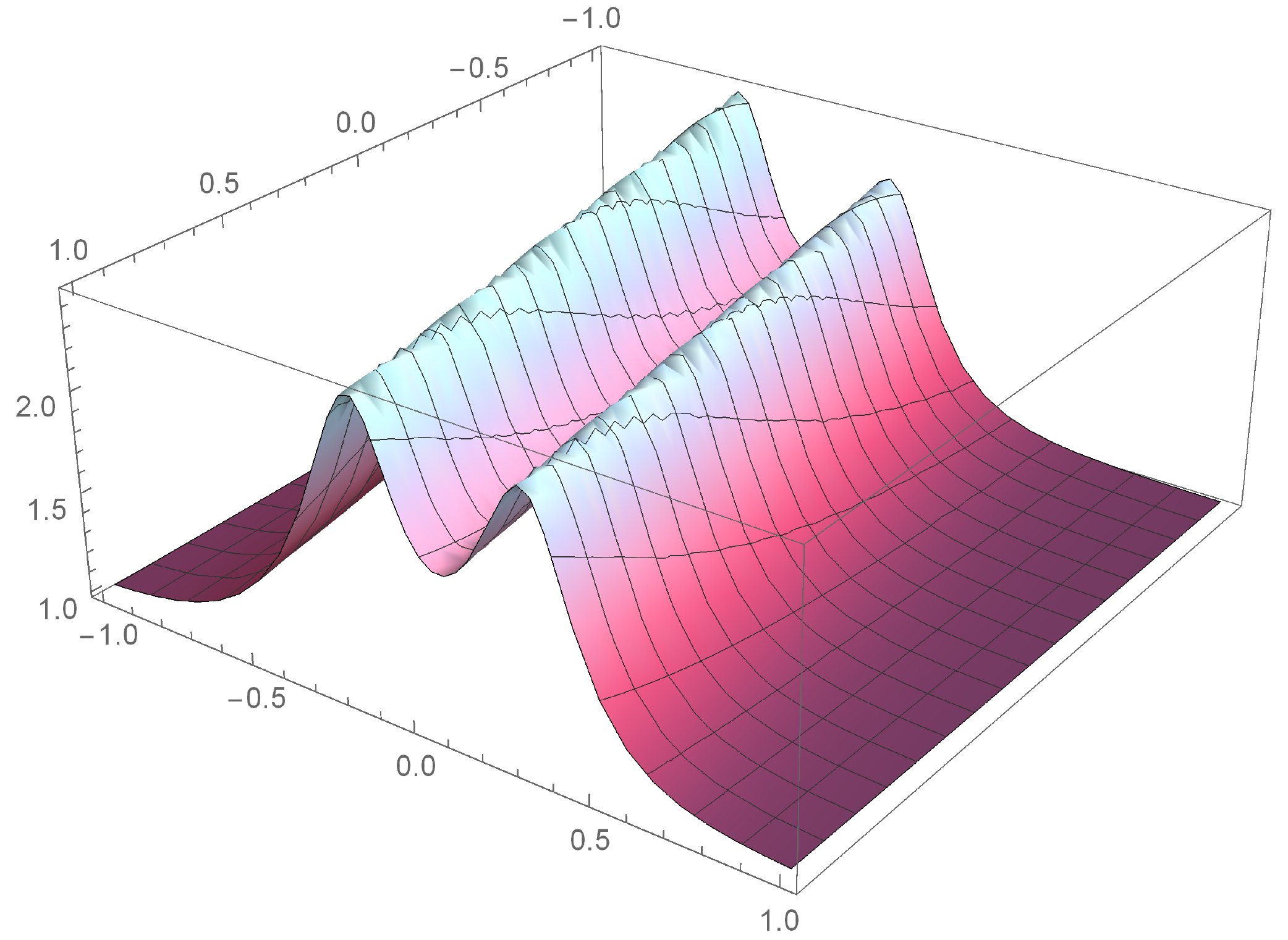}
	}
	\quad
		\subfigure[]{
		\includegraphics[width=0.31\linewidth]{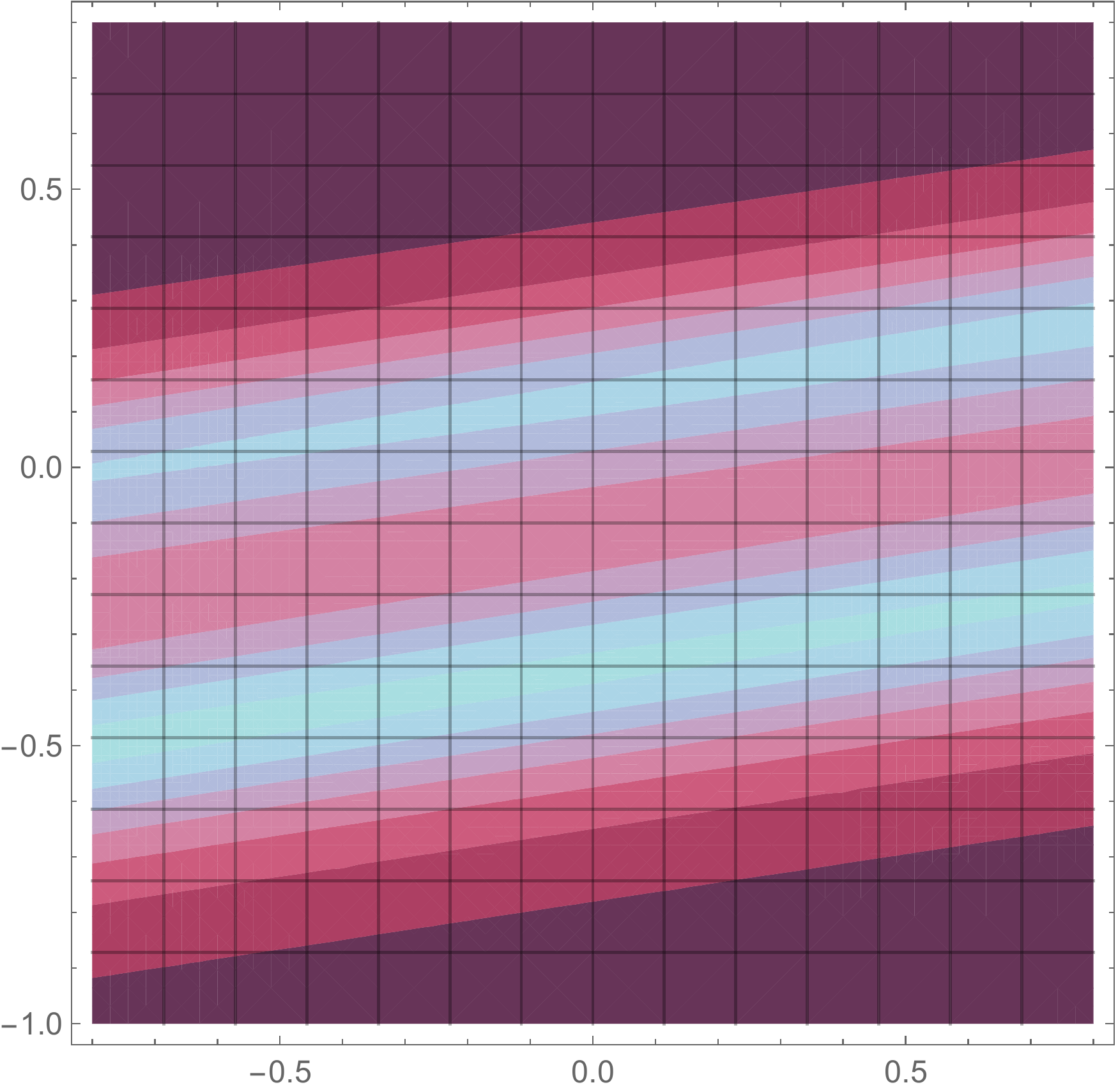}
	}
	\caption{Two-solition solution: (a)  Perspective view of the wave;  (b)  The contour  of the wave.}
\end{figure}
\begin{figure}[H]
	\centering
	\subfigure[]{
		\includegraphics[width=0.41\linewidth]{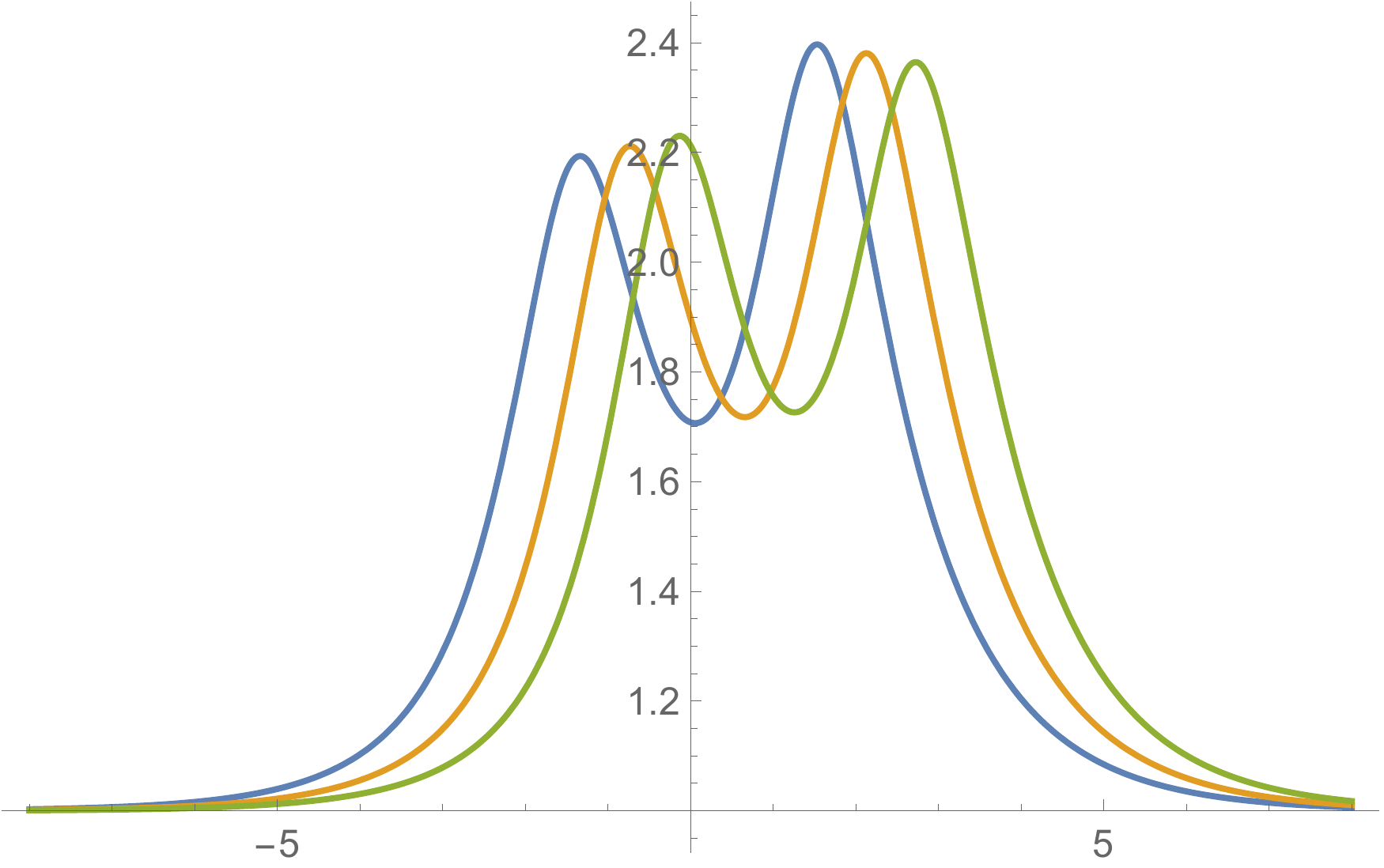}
	}
	\subfigure[]{
		\includegraphics[width=0.41\linewidth]{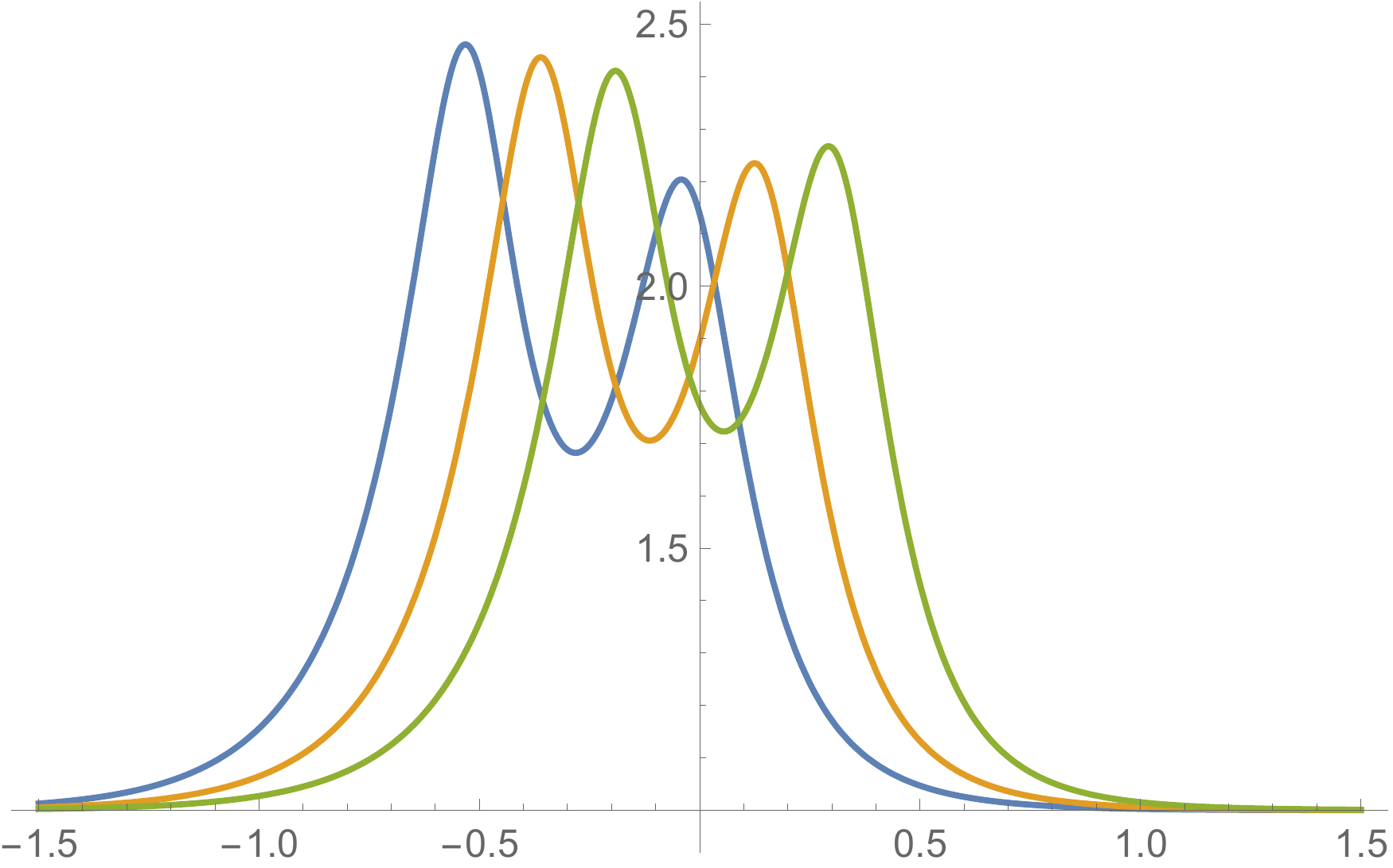}
	}
	\caption{Two-soliton solution: (a) Propagation pattern of the wave along $x$-orientation with $t=-0.2$(blue), $t=0$(orange), $t=0.2$(green); (b)   Propagation pattern  of the wave along $t$-orientation with $x=-1$(blue),  $x=0$(orange),  $x=1$(green).}
	\label{fig:3}
\end{figure}

\noindent {\bf Case II.  $N_1=2$, $N_2=0$:}

Similarly, let $\zeta_j=Z_je^{i\gamma_j}$, with $Z_j>1$ and $\gamma_j \in (\frac{\pi}{2},\pi)$, $j=1,2$, $\zeta_1 \neq \zeta_2$ . Then the other points in discrete spectrum are
\begin{align*}
&-\zeta_1=-Z_1e^{i\gamma_1}, \zeta_1^*=Z_1e^{-i\gamma_1},  -\zeta_1^*=-Z_1e^{-i\gamma_1}, \\
&\zeta_3=-\frac{1}{Z_1}e^{i\gamma_1},  -\zeta_3=\frac{1}{Z_1}e^{i\gamma_1},\zeta_3^*=-\frac{1}{Z_1}e^{-i\gamma_1},  -\zeta_3^*=\frac{1}{Z_1}e^{-i\gamma_1}\\
&-\zeta_2=-Z_2e^{i\gamma_2}, \zeta_2^*=Z_2e^{-i\gamma_2},  -\zeta_2^*=-Z_2e^{-i\gamma_2}, \\
&\zeta_4=-\frac{1}{Z_2}e^{i\gamma_2},  -\zeta_4=\frac{1}{Z_2}e^{i\gamma_2},\zeta_4^*=-\frac{1}{Z_2}e^{-i\gamma_2},  -\zeta_4^*=\frac{1}{Z_2}e^{-i\gamma_2}.
\end{align*}

By using  the theta condition (\ref{thetacondition}), we
have
\begin{equation}
\arg\left( u_-/u_+ \right)+2\nu_0=16\gamma_1+16\gamma_2\nonumber.
\end{equation}
We set $u_-=1$  and  $u_+=e^{i(2\nu_0-16\gamma_1-16\gamma_2)}$.  Let $C_1=e^{\xi_1+i\varphi_1}$, $C_2=e^{\xi_2+i\varphi_2}$, with $\xi_j$, $\varphi_j$ $ \in \mathbb{R}$, $j=1,2$, then $C_3=-\frac{1}{Z_1^2}e^{\xi_1+i(2\gamma_1-\varphi_1)}$, $C_4=-\frac{1}{Z_2^2}e^{\xi_2+i(2\gamma_2-\varphi_2)}$.\\
Substituting  above   data into  formulae   (\ref{refpotential}),  we get the two-soliton solution
\begin{equation}
u(x,t)=e^{2i\nu_-}+2e^{2i\nu_-}\frac{\det\left(\begin{array}{ccccc}
	0 & Y_1 & Y_2 & Y_3 & Y_4\\
	B_1 & 1+A_{11} & A_{12} & A_{13} & A_{14}\\
	B_2 & A_{21} & 1+A_{22} & A_{23} & A_{24}\\
	B_3 & A_{31} & A_{32} & 1+A_{33} & A_{34}\\
	B_4 & A_{41} & A_{42} & A_{43} & 1+A_{44}
	\end{array}\right)}{\det\left(\begin{array}{cccc}
	1+A_{11} & A_{12} & A_{13} & A_{14}\\
	A_{21} & 1+A_{22} & A_{23} & A_{24}\\
	A_{31} & A_{32} & 1+A_{33} & A_{34}\\
	 A_{41} & A_{42} & A_{43} & 1+A_{44}
	\end{array}\right)}. \label{pp}
\end{equation}
where
\begin{align*}
&\theta(x,t,\zeta_j)=-\frac{1}{4\alpha}(\zeta_j^2-\frac{1}{\zeta_j^2})\left[x+(\frac{1}{2\alpha}(\zeta_j^2+\frac{1}{\zeta_j^2})-4\alpha-\frac{2}{\alpha})t\right],\hspace{0.5cm}j=1,2,3,4,\\
&c_j(x,t,z)=\frac{C_j}{z^2-\zeta_j^2}e^{-2i\theta(x,t,\zeta_j)},\hspace{0.5cm}j=1,2,3,4,\\
&B_n=1+2\sum_{j=1}^{4}c_j(\zeta_n^*),\hspace{0.5cm}n=1,2,3,4,\\
&A_{nk}=-\sum_{j=1}^{4}4\zeta_j^2c_j(\zeta_n^*)c_k^*(\zeta_j^*),\hspace{0.5cm}n,k=1,2,3,4,\\
&Y_n=-C_n^*e^{2i\theta(x,t,\zeta_n^*)},\hspace{0.5cm}n=1,2,3,4.
\end{align*}
In this case, the wave propagation  properties  of two-soliton solution (\ref{pp})  are shown in   Figure 11 and Figure 12.
 This kind of  two-soliton wave   exhibits    two-peak  localized wave,   there  are no  oscillatory waves  and  smooth   on  both  peaks;
   While propagation of the    wave   along  $x$-orientation   and  $t$-orientation   are locally   stable   in the middle,
     and  two  tail  parts      go  to a  nonzero constant.

Of course the  $N$-soliton solution expressions  (\ref{refpotential})   are not limited to one-soliton  and two-soliton solutions above, and
it allows us to obtain explicit solutions with an arbitrary number of soliton solutions for the modified NLS equation.  Finally,  we should remak   that
 in the limit $u_0\rightarrow 0$,  both kinds of one-soliton and two-soliton solutions  reduce to the same,    soliton solutions  of the modified  NLS equation with
 zero boundary condition \cite{Doktorov}. \\
\begin{figure}[H]
	\centering
	\subfigure[]{
		\includegraphics[width=5.4cm]{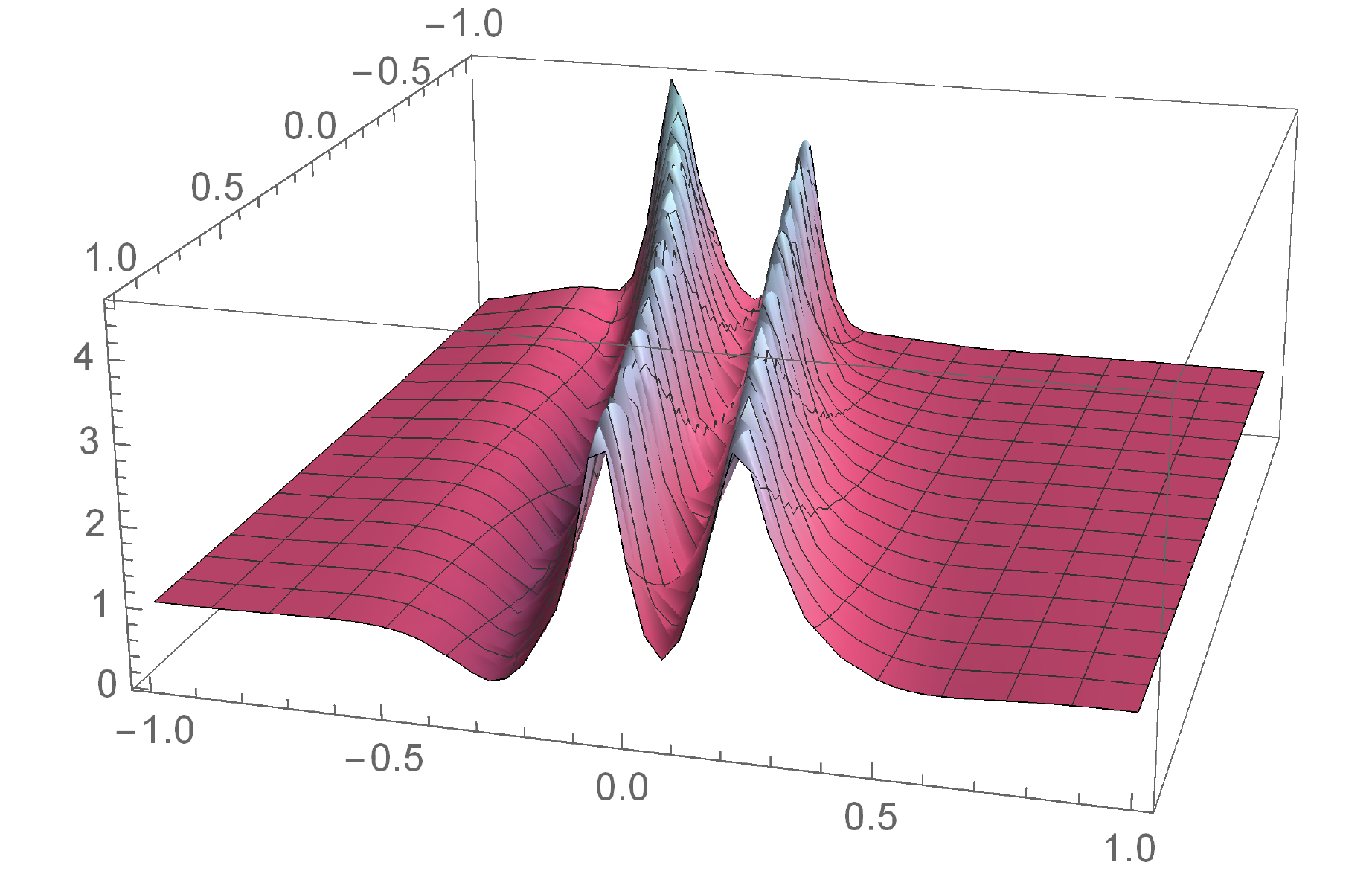}
	}
	\quad
	\subfigure[]{
		\includegraphics[width=0.31\linewidth]{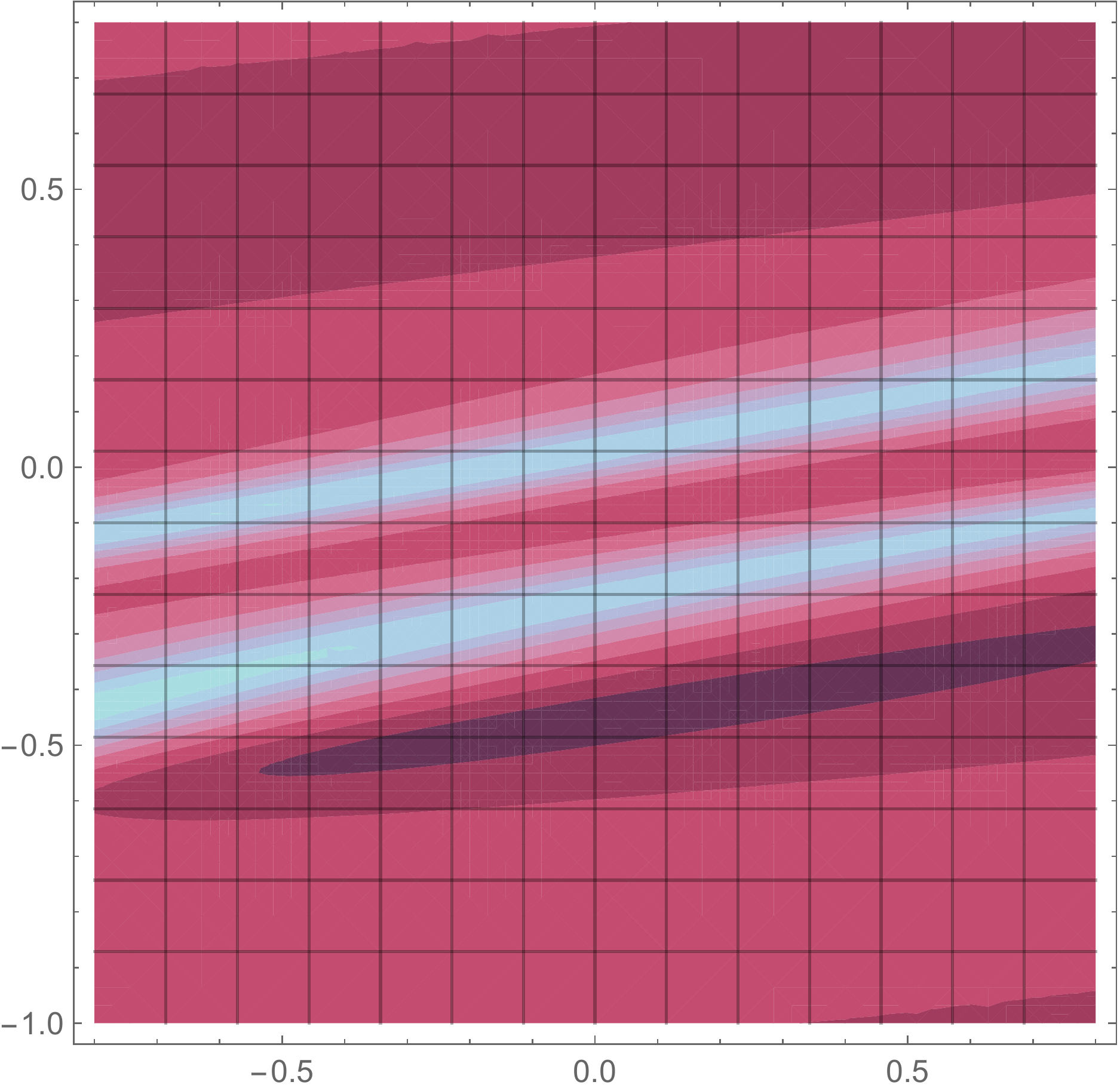}
	}
	\caption{Two-solition solution: (a)  Perspective view of the wave;  (b)  The contour  of the wave.}
\end{figure}
\begin{figure}[H]
	\centering
	\subfigure[]{
		\includegraphics[width=0.41\linewidth]{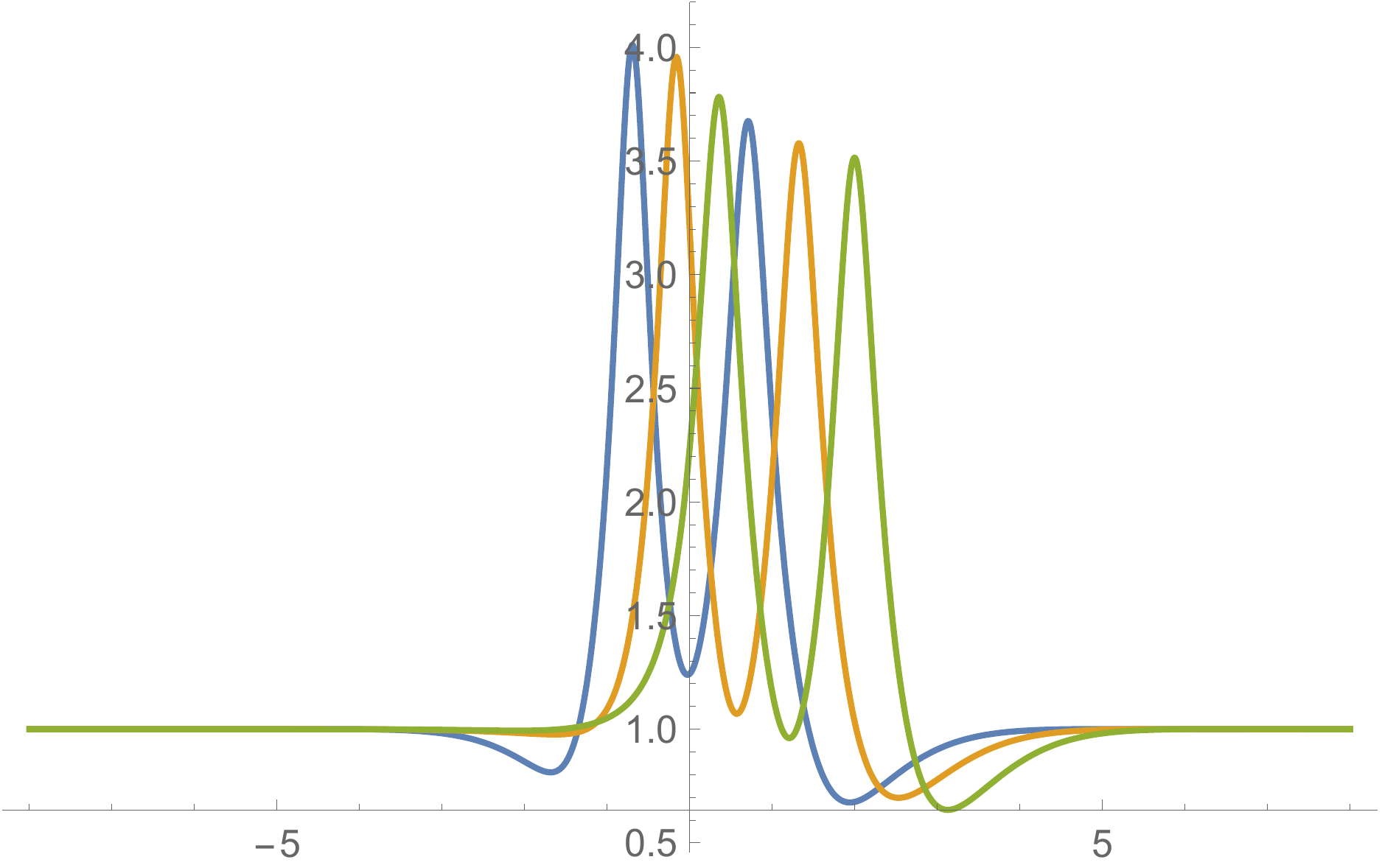}
	}
	\subfigure[]{
		\includegraphics[width=0.41\linewidth]{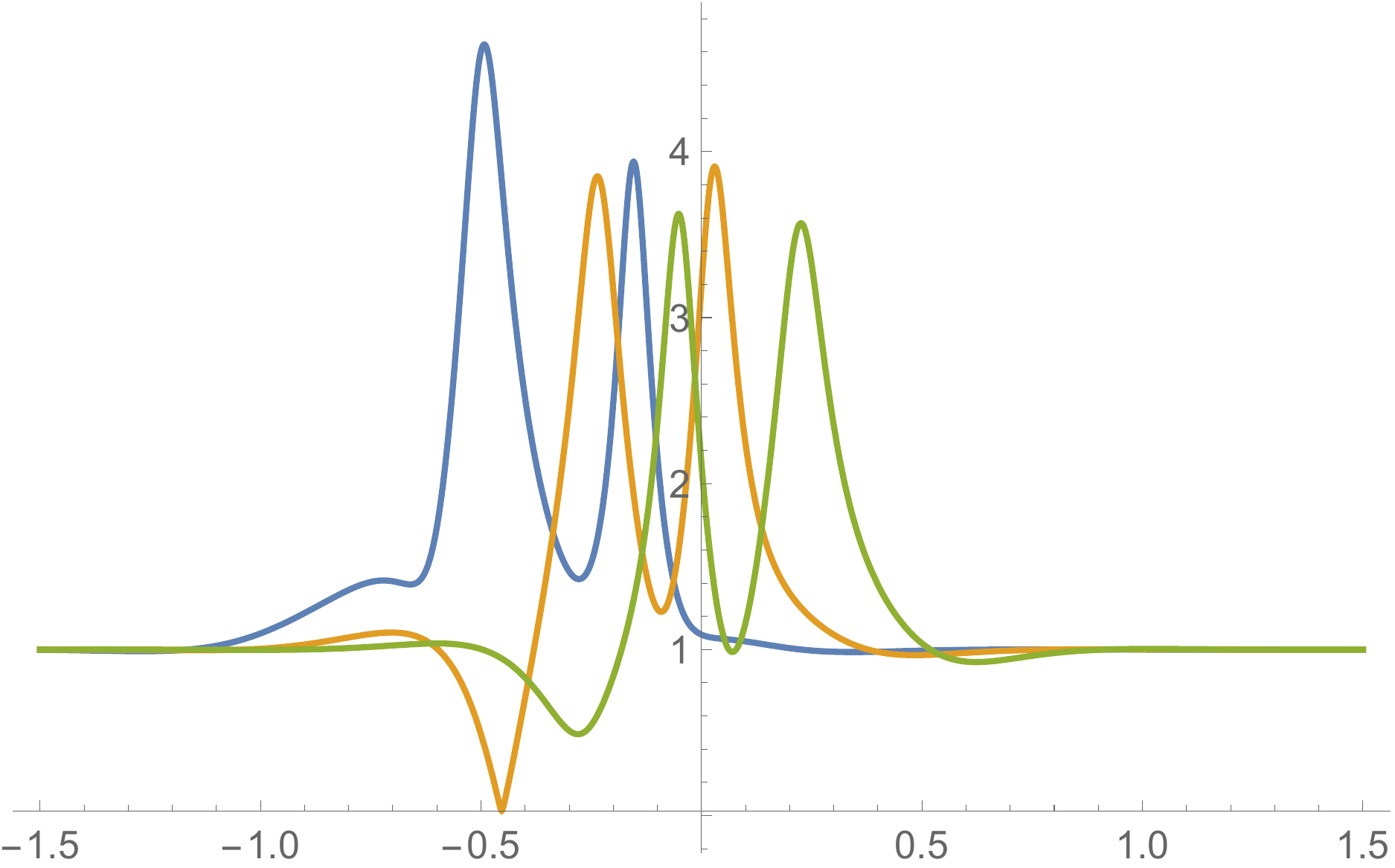}
	}
	\caption{Two-soliton solution: (a) Propagation pattern of the wave along $x$-orientation with $t=-0.2$(blue), $t=0$(orange), $t=0.2$(green); (b)   Propagation pattern  of the wave along $t$-orientation with $x=-1$(blue),  $x=0$(orange),  $x=1$(green).}
	\label{fig:3}
\end{figure}

\noindent\textbf{Acknowledgements}

This work is supported by  the National Science
Foundation of China (Grant No. 11671095,  51879045).

\hspace*{\parindent}
\\

\end{document}